\documentclass[draftcls,onecolumn]{IEEEtran}
\usepackage[T1]{fontenc}
\usepackage[utf8]{inputenc}
\usepackage{cite}
\usepackage{physics}
 \usepackage{tikz}
 \usepackage{pgfplots}
\pgfplotsset{compat=1.18}\usepackage{amsmath}
\usetikzlibrary{arrows.meta,positioning,calc}
\usepackage{amsmath,amssymb,amsthm,mathtools,bm}
\usepackage{bbm}
\usepackage{dsfont}
\usepackage[shortlabels]{enumitem}
\usepackage{xcolor}
\usepackage[hypertexnames=false,hyperfootnotes=false]{hyperref}
\hypersetup{
    colorlinks=true, 
    linkcolor= blue, 
    citecolor=blue 
}
\usepackage[nameinlink,noabbrev]{cleveref}
\usepackage{stmaryrd}
\usepackage{setspace}
\newtheorem{theorem}{Theorem}

\newtheorem{lemma}{Lemma}
\newtheorem{corollary}{Corollary}
\theoremstyle{definition}
\newtheorem{definition}{Definition}

\DeclareMathOperator{\Prb}{\mathbb{P}}

\setlist{nosep}

\allowdisplaybreaks

\usepackage{algorithm} \usepackage{algpseudocode}
\newcounter{algsubstate}

\newenvironment{algsubstates}
  {\setcounter{algsubstate}{0}
   \renewcommand{\State}{
     \refstepcounter{algsubstate}
     \Statex {\footnotesize\alph{algsubstate}:}\space}}
  {}
\newcommand\numberthis{\addtocounter{equation}{1}\tag{\theequation}}
\title{Quantum-State Oblivious Transfer over Quantum Erasure Channels}
\author{R\'{e}mi A. Chou
\thanks{R. A. Chou is with the Department of Computer Science and Engineering,
The University of Texas at Arlington, Arlington, TX 76019 USA
(e-mail: remi.chou@uta.edu).}
}
\begin{document}

\maketitle
\begin{abstract}
We study quantum-state oblivious transfer over the quantum erasure channel under information-theoretic security guarantees. Specifically, Alice holds two quantum systems, Bob wishes to recover one of them, and the parties have access to independent channel uses and authenticated public classical communication, with no initial shared entanglement. We consider both honest-but-curious and malicious behavior. We establish finite-blocklength converse and achievability bounds, which match in the honest-but-curious setting. The malicious setting construction combines interactive hashing with tests of whether Bob received selected channel outputs, and distributes each input across several qubits so that erasures conceal an entire input from a dishonest receiver, even when the inputs are entangled with an external reference. We also derive fixed-error asymptotic expansions and show that   both malicious and honest-but-curious settings have the same capacity.
\end{abstract}

\section{Introduction}

Consider an oblivious transfer setting, where Alice holds two messages, Bob chooses one of them, and the protocol must guarantee that Bob learns only his chosen message while Alice learns nothing about his choice. First introduced in \cite{rabin1981exchange,even1985randomized}, oblivious transfer is a fundamental primitive in cryptography because it is complete for secure two-party computation \cite{kilian1988founding}.

In the information-theoretic setting, oblivious transfer cannot be realized from noiseless communication alone, and additional resources such as noisy channels or correlated randomness are needed \cite{crepeau1988achieving,crepeau2004efficient,imai2006oblivious,ahlswede2013oblivious,nascimento2008oblivious}. In quantum settings, noiseless quantum communication is also insufficient, in general, to obtain oblivious transfer satisfying both correctness and the privacy requirements of Alice and Bob without additional assumptions \cite{lo1997insecurity,salvail2009power,buhrman2012complete}. This motivates us to study   noisy quantum channels as a resource for realizing oblivious transfer.

In this paper, we study a quantum-state version of oblivious transfer over the quantum erasure channel \cite{bennett1997capacities} in honest-but-curious and malicious settings. In the former, the parties follow the prescribed operations but may keep all private side information generated during the protocol.\footnote{Related
retained-view   formulations have been studied in
quantum two-party cryptography
\cite{salvail2009power}.} In the latter, either party may replace its prescribed operations by arbitrary quantum operations, and security is defined through an ideal operation that transfers only one  quantum input. Specifically, Alice holds two $m$-qubit systems $A_1$ and $A_2$, which may be arbitrarily entangled with an inaccessible reference system $R$. Bob has a choice $C\in\{1,2\}$ and wishes to recover $A_C$. The parties have access to $n$ independent uses of the quantum erasure channel   and authenticated public classical communication, but no initial shared entanglement. In each channel use, Alice prepares the quantum input, possibly as a function of her private registers and the public communication, and Bob receives the corresponding channel output. At the end of the protocol, Bob must recover the chosen quantum system $A_C$ exactly, including its correlations with the reference system $R$. At the same time, Alice must not learn Bob's choice $C$, and Bob must obtain no information about the non-chosen quantum system $A_{\bar C}$ beyond what is already contained in its correlations with $A_C$. More specifically, Bob’s choice privacy requires that the public communication must have the same distribution for both values of $C$, and, once the public communication is fixed, Alice’s final state must be identical for both choices.

Our main contribution is a  finite-blocklength converse for both settings, a matching finite-blocklength achievability bound in the honest-but-curious setting, and finite-blocklength bounds and a matching capacity result in the malicious setting. This characterization holds for arbitrary input states, including
states entangled with an external reference, while the privacy
guarantees account for all private quantum side information retained
by the parties. The main challenge in the achievability is to use Bob's knowledge of
which channel uses were erased to protect the non-chosen state without
revealing his choice through the public communication. To address this, the honest-but-curious achievability scheme uses teleportation over entanglement generated through the erasure channel \cite{bennett1993teleporting}. Specifically, Alice first sends halves of maximally entangled pairs through the channel. Bob selects unerased positions for the chosen state and erased positions for the non-chosen state, and publicly labels the two selected sets by the corresponding states. Alice teleports both quantum states using the corresponding positions. The intuition behind the correctness and privacy
guarantees is as follows.  Correctness follows because the positions assigned to the chosen state are unerased. Bob’s choice remains private because the public communication has the same distribution for either choice and, once it is fixed, Alice’s final state does not depend on which positions were erased. Privacy of the non-chosen state follows because the teleportation outcomes are uniform and Bob receives only erasure flags at the positions assigned to that state.

An additional challenge for the achievability in the  malicious setting is that Bob may falsely assign received positions to both states. Our construction combines interactive hashing and tests of his possession of selected outputs, following the approach used for classical erasure-channel oblivious transfer \cite{crepeau2006optimal,pinto2011achieving,dowsley2017oblivious}, with a random Clifford encoding that hides a quantum
input from Bob, up to a small error, when sufficiently
many of its encoded qubits are erased. To prove security against a dishonest Alice or Bob, we compare a real execution of the protocol with an ideal experiment
in which a trusted third party, who receives Alice's two
inputs and Bob's choice, delivers only the chosen input to Bob, without
revealing Bob's choice to Alice. In this ideal experiment, we reproduce the dishonest party's
final private information and the public communication
using only that party's initial information  and any
input received from the trusted third party.
We then show that the resulting joint final state
is close to that of the real execution. 

The intuition behind the converse is as follows. The first part is a Schmidt-rank growth argument across the bipartition between Bob's laboratory and the remaining systems. Since each unerased qubit can increase this Schmidt rank by at most a factor of two, exact recovery of an $m$-qubit state requires at least $m$ unerased channel uses. The second part compares an execution in which Bob requests the first state with one in which he requests the second. A purification argument then shows that, unless at least one channel use is erased per message qubit, Bob cannot be prevented from obtaining information about the state he did not request. Together with our achievability results, the converse establishes the capacity and shows that it is the same in the honest-but-curious and malicious settings.

\emph{Related work.} Our work differs from several related lines of work. Classical oblivious transfer over noisy channels concerns the transfer of classical strings \cite{crepeau1988achieving,imai2006oblivious,ahlswede2013oblivious,nascimento2008oblivious,shekofteh2026capacity,pinto2011achieving,dowsley2017oblivious,suda2025improved,tyagi2015converses}, whereas here the transferred objects are  quantum states that may be entangled with an external reference. This requires correctness and privacy conditions that preserve quantum correlations, not only classical distributions. Our setting is thus also distinct from quantum protocols for oblivious transfer of classical messages \cite{crepeau1988achieving,bennett1991practical,crepeau1994quantum}. This literature includes protocols assuming access to secure bit commitment \cite{yao1995security,mayers1996quantum}, quantitative analyses of imperfect and semi-honest quantum oblivious transfer \cite{chailloux2013lower,chailloux2016optimal,amiri2021imperfect}, and protocols based on bounded or noisy quantum storage assumptions \cite{damgaard2008cryptography,wehner2008cryptography,konig2012unconditional,lupo2023error,furrer2018continuous, erven2014experimental} or computational assumptions \cite{damgaard2009improving,diamanti2026practical,grilo2021oblivious,bartusek2021one}. For a review of quantum protocols for oblivious transfer of classical messages and their security assumptions, see \cite{santos2022quantum}. Beyond oblivious transfer of classical messages, oblivious transfer for quantum messages has also been studied
in two-server models
\cite{hayashi2024two,yang2025quantum,huang2026two}
and in an all-or-nothing formulation
\cite{cai2025counterfactual}. In contrast, we consider a single sender, a single receiver, and a noisy quantum channel as the physical resource, and our goal is a finite-blocklength and capacity characterization. Finally, our definition of security against a dishonest Alice or Bob is related to classical statistical-security conditions \cite{crepeau2008statistical}, quantum frameworks for security and composition \cite{fehr2009composing,unruh2010universally}, and secure two-party evaluation of quantum operations against active adversaries \cite{dupuis2012actively} through the shared requirement that a real execution be reproducible, up to a small error, in an ideal experiment.

The remainder of the paper is organized as follows. Section~\ref{sec:problem_statement} formalizes the problem statement. Section~\ref{secMR} presents our main results.  Sections~\ref{sec:ach} and~\ref{sec:converse} prove the achievability and converse, respectively, for the honest-but-curious setting. Section~\ref{sec:malicious_ach} proves the achievability  for the malicious setting. Finally, Section~\ref{sec:concl} provides concluding remarks.

\section{Notation}

For $n\in\mathbb N$, define $[n]\triangleq[1,n]\cap \mathbb N$.
For a finite set $\mathcal S$, let $|\mathcal S|$
denote its cardinality and $\binom{\mathcal S}{k}$
the collection of its $k$-element subsets.
For indexed variables or registers, define
$X^n\triangleq(X_i)_{i\in[n]}$ and
$X_{\mathcal S}\triangleq(X_i)_{i\in\mathcal S}$.
$\oplus$ denotes addition modulo two.
The logarithms $\log$ and $\ln$ are to base two
and base $e$, respectively.

For a quantum register $A$, let $\mathcal L(A)$
denote the space of linear operators on its
Hilbert space.
Let $I_A$ and $\operatorname{id}_A$ denote the
identity operator and identity channel on $A$,
respectively.
The map $\operatorname{id}_{A\to B}$ denotes
the identity channel with output register labeled $B$.
For an operator $X$, define its trace norm as
$\|X\|_1\triangleq\operatorname{Tr}\sqrt{X^\dagger X}$,
where $\dagger$ denotes the adjoint.
The superscript $T$ denotes the transpose in
the computational basis.
A quantum channel is a completely positive,
trace-preserving (CPTP) map.

For a state $\rho^{AB}$, 
define $\rho^A\triangleq\operatorname{Tr}_B[\rho^{AB}]$.
For a pure state $|\psi\rangle$, write
$\psi\triangleq|\psi\rangle\langle\psi|$.
States are positive semidefinite operators with
unit trace unless explicitly called subnormalized,
in which case their trace is at most one. For a  state $\rho^{ABD}$, define the
von Neumann entropy
$H(A)_\rho\triangleq-\operatorname{Tr}[\rho^A\log\rho^A]$,
the quantum mutual information
$I(A;B)_\rho\triangleq H(A)_\rho+H(B)_\rho-H(AB)_\rho$,
and the quantum conditional mutual information
$I(A;B\mid D)_\rho\triangleq
H(AD)_\rho+H(BD)_\rho-H(D)_\rho-H(ABD)_\rho$, where the state subscript is omitted when clear
from context.
\section{Problem Statement}
\label{sec:problem_statement}

Consider the qubit erasure channel $
\mathcal E_\epsilon(\rho)
=
(1-\epsilon)\rho+\epsilon |e\rangle\langle e|$, 
where $|e\rangle$ is an erasure flag orthogonal to the input
qubit space.  Alice holds two quantum systems and Bob wishes to
recover one of them, using $n$   uses of $\mathcal E_\epsilon$ and
authenticated noiseless public classical communication.
In honest executions, the remaining initial private registers
are independent of the inputs and in a product state across the
parties, i.e., no entanglement or secret correlated randomness is
initially shared.

We formalize such a quantum oblivious transfer
protocol in the following definition. In the honest-but-curious setting, both parties follow the
prescribed protocol. In the malicious setting, at most one party,
fixed before execution, may replace its operations arbitrarily.

\begin{definition} \label{def1}
An $(n,m)$ quantum oblivious transfer protocol
over $\mathcal E_\epsilon$ consists of:
\begin{itemize}
\item Two $m$-qubit message systems $A_1$ and $A_2$, held by Alice,
with joint pure state $\Psi^{A_1A_2R}$, where the reference $R$ is
inaccessible. We also define $\Psi^{A_cR}
\triangleq
\operatorname{Tr}_{A_{\bar c}}
\!\left[\Psi^{A_1A_2R}\right]$,
$c\in\{1,2\}$ 
with $\bar c\triangleq 3-c$.

\item Bob's choice $C$, uniformly distributed over $\{1,2\}$
and independent of $\Psi$, where $C=c$ means that Bob wishes
to obtain $A_c$.

\item $n$ independent uses of $\mathcal E_\epsilon$, each
modeled by the isometric extension $V_\epsilon:A\to BW$, $
V_\epsilon|\psi\rangle^A
=
\sqrt{1-\epsilon}\,|\psi\rangle^B|e\rangle^W
+
\sqrt{\epsilon}\,|e\rangle^B|\psi\rangle^W$,
where $B$ is Bob's output and $W$ is an inaccessible environment.
\end{itemize}

The protocol proceeds over $n$ rounds:
\begin{enumerate}
\item In round $i\in[n]$:
\begin{enumerate}
\item Alice and Bob exchange a finite sequence of authenticated
public classical messages. These messages may depend on their
current private registers and all earlier public messages.

\item Alice prepares one qubit $X_i$ from her current private
register and all public messages exchanged so far, and sends
$X_i$ through $\mathcal E_\epsilon$.

\item Bob receives the channel output $Y_i$, which is added
to his private register.
\end{enumerate}

\item After the $n$-th round, Alice and Bob may exchange a final
finite sequence of authenticated public classical messages.
Let $M_0$ denote the collection of all public messages.

\item Bob applies a final decoding operation depending on $C$.
If $C=c$, this operation produces an output register $B_c$,
a residual private register $E_c$, and a public flag
$F\in\{0,1\}$, where $F=1$ means that Bob aborts.
The overall public communication is
$M\triangleq(M_0,F)$.
\end{enumerate}
Let $\mathsf S$ be the classical register whose value is the
set of erased channel positions,
\begin{equation*}
\mathsf S
\triangleq
\left\{i\in[n]:Y_i=|e\rangle\right\},
\end{equation*}
which is retained by Bob as a subregister of $E_c$.
Let $Q_A$ denote Alice's final private register.
We assume no trusted deletion, i.e., every local operation is implemented isometrically,
with its environment, randomness, and measurement records
retained in $Q_A$ or $E_c$.
\hfill\(\diamond\) \end{definition}

For an honest execution with choice $c\in\{1,2\}$ and input
$\Psi^{A_1A_2R}$, let
$\Omega^{(c)}_{Q_A M B_cE_cRW^n}(\Psi)$ denote the final global
state, and define
\begin{equation*}
\omega^{(c)}_{Q_A M_0 F B_cE_cR}(\Psi)
\triangleq
\operatorname{Tr}_{W^n}
\!\left[
\Omega^{(c)}_{Q_A M B_cE_cRW^n}(\Psi)
\right].
\end{equation*}
For $f\in\{0,1\}$, define
\begin{equation*}
p_f^{(c)}(\Psi)
\triangleq
\operatorname{Tr}
\!\left[
\left(
I_{Q_A M_0 B_cE_cR}\otimes |f\rangle\!\langle f|^F
\right)
\omega^{(c)}_{Q_A M_0 F B_cE_cR}(\Psi)
\right].
\end{equation*}
For any subset of registers
$\mathcal S\subseteq\{Q_A,M_0,B_c,E_c,R\}$, and whenever
$p_f^{(c)}(\Psi)>0$, define
\begin{equation*}
\omega^{(c)}_{\mathcal S\,|\,F=f}(\Psi)
\triangleq
\frac{
\operatorname{Tr}_{\mathcal S^c F}
\!\left[
\left(
I_{Q_A M_0 B_cE_cR}\otimes |f\rangle\!\langle f|^F
\right)
\omega^{(c)}_{Q_A M_0 F B_cE_cR}(\Psi)
\right]
}{
p_f^{(c)}(\Psi)
},
\end{equation*}
where $\mathcal S^c$ denotes the complement of $\mathcal S$
inside $\{Q_A,M_0,B_c,E_c,R\}$.
Conditional states given $M_0=m_0$, $F=f$, and $\mathsf S=s$
are defined analogously whenever the corresponding joint
probability is positive.

\begin{definition}
\label{def:qot_delta_achievability}
Fix $n\in\mathbb N$, $m\in\mathbb Z_{\geq0}$,
$\delta\in[0,1[$.
The message size $m$ is $\delta$-achievable in the
honest-but-curious setting if an $(n,m)$ protocol satisfies
Conditions 1--4 below.

\noindent\emph{1. Abort probability.}
\begin{equation}
\sup_{\Psi,c}p_1^{(c)}(\Psi)\leq\delta,
\label{eq:qot_abort}
\end{equation}
where the supremum is over all $c\in\{1,2\}$ and all pure
states $\Psi^{A_1A_2R}$.

\medskip
\noindent\emph{2. Correctness.}
For every $c\in\{1,2\}$ and every input state
$\Psi^{A_1A_2R}$, if $p_0^{(c)}(\Psi)>0$, then
\begin{equation}
\omega^{(c)}_{B_cR\,|\,F=0}(\Psi)
=
(\operatorname{id}_{A_c\to B_c}\otimes\operatorname{id}_R)
\!\left(\Psi^{A_cR}\right).
\label{eq:qot_correctness}
\end{equation}

\medskip
\noindent\emph{3. Receiver's choice privacy.}
For every input $\Psi^{A_1A_2R}$, every $m_0$, and every
$f\in\{0,1\}$,
\begin{equation}
\Prb[M_0=m_0,F=f\mid C=1]
=
\Prb[M_0=m_0,F=f\mid C=2],
\label{eq:qot_receiver_privacy1}
\end{equation}
and there exists a state $\tau_{Q_A Rm_0f}(\Psi)$ such that,
for every $c\in\{1,2\}$ and every erased set $s$ for which
$\Prb[M_0=m_0,F=f,\mathsf S=s\mid C=c]>0$,
\begin{equation}
\omega^{(c)}_{Q_A R
\mid M_0=m_0,\,F=f,\,\mathsf S=s}(\Psi)
=
\tau_{Q_A Rm_0f}(\Psi).
\label{eq:qot_receiver_privacy2}
\end{equation}

\medskip
\noindent\emph{4. Non-chosen state privacy.}
For each $c\in\{1,2\}$, there exists a CPTP map $\mathsf{Sim}_{c,n}:
\mathcal L(A_c)\longrightarrow\mathcal L(B_cE_cM_0F)$, 
independent of $\Psi$, such that, for every input
$\Psi^{A_1A_2R}$,
\begin{equation}
\omega^{(c)}_{B_cE_cM_0FR}(\Psi)
=
(\mathsf{Sim}_{c,n}\otimes\operatorname{id}_R)
\!\left(\Psi^{A_cR}\right).
\label{eq:qot_global_destructive_privacy_general}
\end{equation}
\hfill\(\diamond\) \end{definition}
\medskip
\eqref{eq:qot_abort} ensures that the protocol aborts with probability at most $\delta$, uniformly over Alice's input and Bob's choice. \eqref{eq:qot_correctness}~requires Bob, upon success, to recover the chosen state exactly, including all its correlations with the reference. \eqref{eq:qot_receiver_privacy1} ensures that the public communication does not reveal Bob's choice and that, once this communication is fixed,  \eqref{eq:qot_receiver_privacy2}  ensures that Alice’s final state depends neither on his choice nor on which channel uses were erased, so that  privacy holds for each realization of the erasure channel, rather than only after averaging over its random erasures.  \eqref{eq:qot_global_destructive_privacy_general} requires Bob's complete view   to be reproducible from the chosen state alone, ensuring that he obtains no additional information about the non-chosen state.

\begin{definition}
\label{def:qot_delta_achievabilityb}
Fix $n\in\mathbb N$, $m\in\mathbb Z_{\geq0}$,
$\delta\in[0,1[$, and $\nu\in[0,1]$.
The message size $m$ is $\delta$-achievable in the malicious
setting if an $(n,m)$ protocol  satisfies Conditions 1--4 from Definition \ref{def:qot_delta_achievability} and Condition 5 below.

\medskip
\noindent\emph{5. Security against malicious users.}
In the protocol execution, let $Z$ and $V$ denote the corrupted party's initial and final
private registers, and let $M=(M_0,F)$ be the public communication.
We compare the real execution with an ideal experiment in which
a simulator $\mathsf{Sim}$ replaces the corrupted party, receives $Z$,
and interacts with a trusted genie as follows:

\begin{itemize}
\item \emph{Corrupted Bob.}
The initial state $\rho^{A_1A_2ZR}$ is arbitrary. Using $Z$, the simulator  chooses 
$C'\in\{1,2\}$. The genie transfers only $A_{C'}$ to the simulator, which  may abort before or after this transfer. The simulator generates $V,M$ using only $Z$ and,
if received, $A_{C'}$.  All undelivered inputs are discarded.

\item \emph{Corrupted Alice.}
The register $Z$ contains Alice's inputs, and
$\rho^{CZR}
=
\frac12\sum_{c=1}^2
|c\rangle\!\langle c|^C\otimes\rho_c^{ZR}$,
for arbitrary states $\rho_c^{ZR}$.
Without access to $C$, the simulator jointly generates two $m$-qubit inputs
$A'_1,A'_2$ and the outputs $V,M$ through a CPTP map
$\mathcal A:\mathcal L(Z)
\longrightarrow\mathcal L(A'_1A'_2VM)$.
The genie transfers $A'_C$ to Bob's output register
$B_{\mathrm{out}}$ if $F=0$  and a fixed orthogonal abort symbol $\bot$ if $F=1$.
All undelivered inputs are discarded.
\end{itemize}
To compare the real execution and the ideal experiment, let $B_{\mathrm{out}}$ denote
honest Bob's declared output in the real execution: $B_c$ on
success and $\bot$ on abort.
The comparison is on $HVMR$, where
$H\triangleq CB_{\mathrm{out}}$ when Alice is corrupted and
$H\triangleq\varnothing$ when Bob is corrupted.

For every strategy $\mathcal D$ of the corrupted party, there must exist a simulator
$\mathsf{Sim}_{\mathcal D}$ as above such that, for every initial
state $\rho$ specified above, the normalized final  states
$\rho_{\mathrm{real}}^{HVMR}(\rho,\mathcal D)$ and
$\rho_{\mathrm{ideal}}^{HVMR}(\rho,\mathsf{Sim}_{\mathcal D})$
satisfy
\begin{equation}
\frac12
\left\|
\rho_{\mathrm{real}}^{HVMR}(\rho,\mathcal D)
-
\rho_{\mathrm{ideal}}^{HVMR}(\rho,\mathsf{Sim}_{\mathcal D})
\right\|_1
\leq\nu.
\label{eq:qot_malicious_simulation}
\end{equation}
\hfill\(\diamond\) \end{definition}

\eqref{eq:qot_malicious_simulation} ensures that, even when
a party cheats, the protocol behaves, up to the prescribed
error, as if a trusted party delivered only the state selected
by Bob, did not reveal Bob's choice to Alice, and allowed either
party to abort.
\begin{definition}
\label{def:finite_length_message_size}
For $n\in\mathbb N$, $\delta\in[0,1[$, and $\nu\in[0,1]$,
define the honest-but-curious maximum message size
\begin{equation*}
M_{\mathrm{QOT}}(\epsilon,n,\delta)
\triangleq
\max\left\{
m\in\mathbb Z_{\geq0}:
m\text{ is $\delta$-achievable at blocklength $n$ over }
\mathcal E_\epsilon
\right\}.
\end{equation*}
Similarly, let
$M_{\mathrm{QOT}}^{\mathrm{mal}}(\epsilon,n,\delta,\nu)$
be the largest $(\delta,\nu)$-achievable message size in the
malicious setting.
\hfill\(\diamond\) \end{definition}

\begin{definition}
\label{def:asymptotic_capacity}
A rate $r\geq0$ is achievable in the honest-but-curious setting
if, for every $\delta\in]0,1[$ and all sufficiently large $n$,
the message size $\lfloor nr\rfloor$ is $\delta$-achievable
at blocklength $n$.
The supremum of all achievable rates is called the quantum
oblivious transfer capacity over $\mathcal E_\epsilon$,
and is denoted by $C_{\mathrm{QOT}}(\epsilon)$.

In the malicious setting, $r$ is achievable if, for every
$\delta,\nu\in]0,1[$ and all sufficiently large $n$, the
message size $\lfloor nr\rfloor$ is $(\delta,\nu)$-achievable.
The corresponding supremum is denoted by
$C_{\mathrm{QOT}}^{\mathrm{mal}}(\epsilon)$.
\hfill\(\diamond\) \end{definition}

\section{Main Results} \label{secMR}
We characterize the maximum message size, derive fixed-error asymptotic expansions, and determine the capacity in the honest-but-curious setting in Section~\ref{secA}. We then establish finite-blocklength bounds, derive corresponding asymptotic expansions, and determine the capacity in the malicious setting in Section~\ref{sec:malicious_results}.
\subsection{Honest-but-Curious Setting} \label{secA}

\begin{theorem}\label{thm:finite_length_achievability} Consider the qubit erasure channel $\mathcal E_\epsilon$ with $0<\epsilon<1$. Fix $n\in\mathbb N$, $0<\delta<1$. Let
$\operatorname{Bin}(n,\epsilon)$ denote a binomial random variable with
parameters $n$ and $\epsilon$. Then, we have
\begin{equation*}
\begin{aligned}
 M_{\mathrm{QOT}}(\epsilon,n,\delta)=\max\left\{m\in\{0,\ldots,\lfloor n/2\rfloor\}:\Prb\left[m\leq\operatorname{Bin}(n,\epsilon)\leq n-m\right]\geq1-\delta\right\}.
\end{aligned}
\end{equation*} \end{theorem}
\begin{proof}
     See Section~\ref{sec:ach} for the achievability and Section~\ref{sec:converse} for the  converse.
\end{proof}

In Theorem~\ref{thm:finite_length_achievability},
$\operatorname{Bin}(n,\epsilon)$ represents the number of erased channel
uses. The inequality
$
m\leq\operatorname{Bin}(n,\epsilon)
$
requires at least $m$ erased uses, while
$
\operatorname{Bin}(n,\epsilon)\leq n-m
$
requires at least $m$ unerased uses. The
achievability scheme succeeds whenever both requirements hold. The
converse shows that a secure protocol must abort if either requirement
fails: exact correctness and the absence of initial shared
entanglement require at least $m$ unerased uses, while correctness
together with the privacy conditions requires at least $m$ erased
uses.

From Theorem \ref{thm:finite_length_achievability}, we deduce the following corollary, which gives the honest-but-curious asymptotic expansion for any fixed abort probability.

\begin{corollary}
\label{cor:second_order_fixed_delta_achievability}
Consider the qubit erasure channel $\mathcal E_\epsilon$ with
$0<\epsilon<1$. Fix $0<\delta<1$, and let $\Phi$ denote the CDF of a
standard normal random variable. As $n\to\infty$, we have
\begin{equation*}
M_{\mathrm{QOT}}(\epsilon, n,\delta)
=
\begin{cases}
\displaystyle
n\epsilon
+
\sqrt{n\epsilon(1-\epsilon)}\,\Phi^{-1}(\delta)
+
O(1),
& 0<\epsilon<\frac12,\\[3ex]
\displaystyle
\frac n2
+
\frac{\sqrt n}{2}\,\Phi^{-1}\!\left(\frac{\delta}{2}\right)
+
O(1),
& \epsilon=\frac12,\\[3ex]
\displaystyle
n(1-\epsilon)
+
\sqrt{n\epsilon(1-\epsilon)}\,\Phi^{-1}(\delta)
+
O(1),
& \frac12<\epsilon<1 .
\end{cases}
\end{equation*}
\end{corollary}

\begin{proof} See Appendix \ref{app_sec_ac}. \end{proof}

Corollary \ref{cor:second_order_fixed_delta_achievability} yields  the quantum oblivious transfer
capacity over $\mathcal E_\epsilon$.

\begin{corollary} \label{cor:qot_capacity} Consider the qubit erasure channel $\mathcal E_\epsilon$ with $0<\epsilon<1$. Then, we have \begin{equation*} C_{\mathrm{QOT}}(\epsilon) = \min\{\epsilon,1-\epsilon\}. \end{equation*} \end{corollary}

The capacity equals that of classical oblivious transfer over the binary erasure channel \cite{ahlswede2013oblivious}. In both settings, unerased uses enable recovery of the chosen message, while erased uses protect the non-chosen message, so the smaller fraction determines the capacity.

\subsection{Malicious Setting}
\label{sec:malicious_results}

The next theorem gives a finite-blocklength lower bound on the
maximum message size in the malicious setting.

\begin{theorem}
\label{thm:malicious_finite}
Consider the qubit erasure channel $\mathcal E_\epsilon$ with
$0<\epsilon<1$. Fix $n\in\mathbb N$, $0<\delta,\nu<1$.
Then, we have
\begin{equation*}
M_{\mathrm{QOT}}^{\mathrm{mal}}(\epsilon,n,\delta,\nu)
\geq
\max\left\{
0,\,
\left\lfloor
n\min\{\epsilon,1-\epsilon\}
-
16 \sqrt{
n\ln\!\left(\frac{256}{\min\{\delta,\nu\}}\right)
}
\right\rfloor
\right\}.
\end{equation*}
\end{theorem}

\begin{proof}
See Section~\ref{sec:malicious_ach}.
\end{proof}

The following corollary describes the asymptotic behavior at
fixed abort probability and malicious simulation error.

\begin{corollary}
\label{cor:malicious_asymptotics}
Consider $\mathcal E_\epsilon$ with $0<\epsilon<1$.
Fix $0<\delta,\nu<1$. As $n\to\infty$, we have
\begin{equation}
M_{\mathrm{QOT}}^{\mathrm{mal}}(\epsilon,n,\delta,\nu)
=
n\min\{\epsilon,1-\epsilon\}+O(\sqrt n). \label{eq:mal_matching_expansion1}
\end{equation}
\end{corollary}

\begin{proof}
\eqref{eq:mal_matching_expansion1} follows from
Theorem~\ref{thm:malicious_finite} and
Corollary~\ref{cor:second_order_fixed_delta_achievability} using that $M_{\mathrm{QOT}}^{\mathrm{mal}}(\epsilon,n,\delta,\nu)
\leq M_{\mathrm{QOT}}(\epsilon,n,\delta)$.
\end{proof}

We deduce that both honest-but-curious and malicious settings have the same capacity, as stated next.

\begin{corollary}
\label{cor:malicious_capacity}
For the qubit erasure channel $\mathcal E_\epsilon$ with
$0<\epsilon<1$,
\begin{equation*}
C_{\mathrm{QOT}}^{\mathrm{mal}}(\epsilon)
=
C_{\mathrm{QOT}}(\epsilon)
=
\min\{\epsilon,1-\epsilon\}.
\end{equation*}
\end{corollary}

\section{Achievability proof of Theorem \ref{thm:finite_length_achievability}} \label{sec:ach}

We first review the notion of teleportation \cite{bennett1993teleporting} in Section \ref{sec:teleportation}, then present our achievability scheme and its analysis in Sections \ref{sec:cs} and \ref{sec:csa}, respectively. 

\subsection{Preliminaries} \label{sec:teleportation}

Let $X=X_1\cdots X_m$ be an $m$-qubit system held by a sender, let $R'$ be an arbitrary
reference system, and let $\rho^{XR'}$ be an arbitrary state. Let
$U^m=U_1\cdots U_m$ and $V^m=V_1\cdots V_m$ be two $m$-qubit systems, held by the sender and a receiver, respectively, and prepared in the maximally entangled state,
$\Phi_{2^m}^{U^mV^m}
\triangleq
\ket{\Phi_{2^m}}\!\bra{\Phi_{2^m}}^{U^mV^m}$ with 
$\ket{\Phi_{2^m}}^{U^mV^m}
\triangleq
\frac{1}{\sqrt{2^{m}
}}
\sum_{x\in\{0,1\}^m}
\ket{x}^{U^m}\ket{x}^{V^m}$.
For $a,b\in\{0,1\}^m$, define
$\Phi_{a,b}^{XU^m}
\triangleq
\ket{\Phi_{a,b}}\!\bra{\Phi_{a,b}}^{XU^m}$ with $\ket{\Phi_{a,b}}^{XU^m}
\triangleq
\left(I^{X}\otimes P_{a,b}\right)
\ket{\Phi_{2^m}}^{XU^m}$,
where $\ket{\Phi_{2^m}}^{XU^m}
\triangleq
\frac{1}{\sqrt{2^{m}
}}
\sum_{x\in\{0,1\}^m}
\ket{x}^{X}\ket{x}^{U^m}$ and $P_{a,b}
\triangleq
\bigotimes_{k=1}^m X^{a_k}Z^{b_k}$ is the $m$-qubit Pauli operator acting on $U^m$.

The following lemma provides   the teleportation results \cite{bennett1993teleporting} needed for the achievability proof.
\begin{lemma}[{\cite{bennett1993teleporting}}]
\label{lem:teleportation_properties_ach}
Consider the Bell measurement $
\left\{
\Phi_{a,b}^{XU^m}:
a,b\in\{0,1\}^m
\right\}$ on $XU^m$ and let $T$ be the classical register that stores the measurement
outcome $(a,b)\in\{0,1\}^m\times\{0,1\}^m$. The following holds.

\begin{enumerate}
    \item $T$ is uniform and independent of the
reference $R'$:
\begin{equation}
    \rho^{TR'}
    =
    \frac{1}{{2^{2m}
}}
   \sum_{a,b\in\{0,1\}^m}
\ket{a,b}\!\bra{a,b}^{T}
    \otimes
    \rho^{R'} . \label{eq:tele1}
    \end{equation}

    \item When $T=(a,b)$, the normalized state on the receiver's system $V^m$ and the reference $R'$ is
    \begin{equation}
    P_{a,b}
    \left[
    \left(
    \operatorname{id}_{X\to V^m}\otimes \operatorname{id}_{R'}
    \right)
    \left(\rho^{XR'}\right)
    \right]
    P_{a,b}^\dagger, \label{eq:tele}
    \end{equation}
    where $P_{a,b}$ acts on $V^m$. Consequently, if the receiver
    applies $P_{a,b}^\dagger$ to $V^m$, it obtains
    \begin{equation*}
    \left(
    \operatorname{id}_{X\to V^m}\otimes \operatorname{id}_{R'}
    \right)
    \left(\rho^{XR'}\right).
    \end{equation*}

\end{enumerate}
\end{lemma}

\subsection{Coding Scheme} \label{sec:cs}

Fix $n\in\mathbb N$ and $\delta\in[0,1[$. Define
\begin{equation}
m_{n,\delta}
\triangleq
\max\left\{
m\in\{0,1,\ldots,\lfloor n/2\rfloor\}:
\Prb\left[m\leq \operatorname{Bin}(n,\epsilon)\leq n-m\right]\geq 1-\delta
\right\}.
\label{eq:finite_length_m_def}
\end{equation}
For convenience, we write $m=m_{n,\delta}$. Let
$\Psi^{A_1A_2R}$ be an arbitrary pure input state.
Alice holds the two $m$-qubit systems $A_1$ and $A_2$, while the
reference system $R$ is inaccessible. The protocol is described in
Algorithm~\ref{alg:qot_achievability}.

\begin{algorithm}[t]
\caption{Achievability scheme for quantum oblivious transfer over
$\mathcal E_\epsilon$ (honest-but-curious setting)}
\label{alg:qot_achievability}
\begingroup
\setstretch{1}
\setlength{\abovedisplayskip}{3pt}
\setlength{\belowdisplayskip}{3pt}
\setlength{\abovedisplayshortskip}{2pt}
\setlength{\belowdisplayshortskip}{2pt}

\begin{algorithmic}[1]
\Require Blocklength $n$, target abort probability $\delta$,
$m=m_{n,\delta}$, an input state $\Psi^{A_1A_2R}$, and Bob's choice
$C\in\{1,2\}$.

\State \textit{Entanglement generation.}
\begin{algsubstates}

    \State For each $i\in[n]$, Alice prepares
    $$
    \ket{\Phi^+}^{U_iV_i}
    =
    \frac{1}{\sqrt{2}}(\ket{00}+\ket{11}).
    $$

    \State For each $i\in[n]$, Alice keeps $U_i$ and sends $V_i$
    through $\mathcal E_\epsilon$.

\end{algsubstates}

\State \textit{Abort test.}
\begin{algsubstates}

    \State Bob defines
    $$
    \mathcal{G}
    \triangleq
    \{i\in[n]:\text{the }i\text{-th output is not erased}\}
    \text{ and }
    \mathcal{B}
    \triangleq
    [n]\setminus\mathcal{G}.
    $$

    \State Bob sends to Alice 
$$
F
\triangleq
\mathds{1}\{|\mathcal G|<m\text{ or }|\mathcal B|<m\}.
$$
If $F=1$, Bob sets $B_C$ to a fixed $m$-qubit dummy
state, retains all his remaining registers in $E_C$, and the
protocol terminates. If $F=0$, the protocol continues.

\end{algsubstates}

\State \textit{Subset selection.}
\begin{algsubstates}

    \State If $F=0$, Bob samples uniformly
    $\mathcal G^\star\subseteq\mathcal G$ and
    $\mathcal B^\star\subseteq\mathcal B$ such that
    $|\mathcal G^\star|=|\mathcal B^\star|=m$.

    \State Bob sets
    \begin{equation*}
    (L_1,L_2)
    \triangleq
    \begin{cases}
    (\mathcal G^\star,\mathcal B^\star),&\text{if }C=1,\\
    (\mathcal B^\star,\mathcal G^\star),&\text{if }C=2,
    \end{cases}
    \end{equation*}
    so that, with $\bar C\triangleq3-C$,
    $L_C\subseteq\mathcal G$ and
    $L_{\bar C}\subseteq\mathcal B$.

    \State Bob sends to Alice $(L_1,L_2)$.

\end{algsubstates}

\State \textit{Teleportation.}
\begin{algsubstates}

    \State For each $j\in\{1,2\}$, write the elements of $L_j$ in
    increasing order as
    $L_j=\{\ell_{j,1},\ldots,\ell_{j,m}\}$.

    \State Alice teleports $A_j$ using
    $U_{L_j}\triangleq(U_{\ell_{j,1}},\ldots,U_{\ell_{j,m}})$ as in
    Section~\ref{sec:teleportation}. Specifically, for $k\in[m]$,
    Alice performs a Bell measurement between the $k$-th qubit of
    $A_j$ and $U_{\ell_{j,k}}$. Let
    $T_j=(T_{j,1},\ldots,T_{j,m})\in\{0,1\}^{2m}$ be the
    Bell-measurement outcome, where $T_j$ specifies the $m$-qubit
    Pauli operator
    $P_{T_j}\triangleq\bigotimes_{k=1}^mP_{T_{j,k}}$.

    \State Alice sends to Bob $(T_1,T_2)$.

\end{algsubstates}

\State \textit{Decoding.}
\begin{algsubstates}

    \State Since $L_C\subseteq\mathcal G$, Bob has the
    channel outputs $V_{L_C}$ and applies $P_{T_C}^{\dagger}$ to
    $V_{L_C}$.

    \State Bob outputs the resulting system as $B_C$.

\end{algsubstates}

\end{algorithmic}
\endgroup
\end{algorithm}

We write the complete public transcript excluding $F$ as
$
M_0=(L_1,L_2,T_1,T_2)
$
when $F=0$, and as
$
M_0=\bot
$
when $F=1$.

\subsection{Coding Scheme Analysis} \label{sec:csa}

\subsubsection{Abort probability}

Since the number of erased channel uses follows a binomial distribution
with parameters $(n,\epsilon)$, by \eqref{eq:finite_length_m_def} and
the abort test in Line 2 of Algorithm~\ref{alg:qot_achievability}, we have
$p_1^{(c)}(\Psi)\leq\delta$ for $c\in\{1,2\}$. Since the abort event
does not depend on $\Psi$ or $c$, we have \eqref{eq:qot_abort}.

\subsubsection{Correctness}

Fix $c\in\{1,2\}$ and condition on $F=0$. By Line 3
of Algorithm~\ref{alg:qot_achievability},
$L_c\subseteq\mathcal G$ and $L_{\bar c}\subseteq\mathcal B$.
Hence, Bob possesses $V_{L_c}$, which is used as the receiving system
in the teleportation of $A_c$, while the systems sent at the positions
in $L_{\bar c}$ are routed to the environment. 

Applying Lemma~\ref{lem:teleportation_properties_ach} with
$X=A_c$, $U^m=U_{L_c}$, $V^m=V_{L_c}$, and
$
R'=A_{\bar c}R\,U_{L_{\bar c}}W_{L_{\bar c}},
$
for the Bell-measurement outcome $T_c=t_c$, the normalized state on
$
V_{L_c}A_{\bar c}R\,U_{L_{\bar c}}W_{L_{\bar c}}
$
is
\begin{align*}
&\rho_{t_c}^{V_{L_c}A_{\bar c}R\,U_{L_{\bar c}}W_{L_{\bar c}}}
=
P_{t_c}
\left[
\left(
\operatorname{id}_{A_c\to V_{L_c}}
\otimes
\operatorname{id}_{A_{\bar c}R\,U_{L_{\bar c}}W_{L_{\bar c}}}
\right)
\left(
\Psi^{A_cA_{\bar c}R}
\otimes
\Phi_{2^m}^{U_{L_{\bar c}}W_{L_{\bar c}}}
\right)
\right]
P_{t_c}^{\dagger}.
\end{align*}
By Lines 4 and 5 of Algorithm \ref{alg:qot_achievability}, Alice sends $T_c=t_c$ to Bob, and Bob applies
$P_{t_c}^{\dagger}$ to $V_{L_c}$, giving
\begin{align}
\left(
P_{t_c}^{\dagger}
\otimes
I_{A_{\bar c}R\,U_{L_{\bar c}}W_{L_{\bar c}}}
\right)
\rho_{t_c}^{V_{L_c}A_{\bar c}R\,U_{L_{\bar c}}W_{L_{\bar c}}}
\left(
P_{t_c}
\otimes
I_{A_{\bar c}R\,U_{L_{\bar c}}W_{L_{\bar c}}}
\right)
=
\left(
\operatorname{id}_{A_c\to B_c}
\otimes
\operatorname{id}_{A_{\bar c}R}
\right)
\left(\Psi^{A_cA_{\bar c}R}\right)
\otimes
\Phi_{2^m}^{U_{L_{\bar c}}W_{L_{\bar c}}}.
\label{eq:correctness_selected_after_correction}
\end{align}

In the preceding application of Lemma~\ref{lem:teleportation_properties_ach}, only the Bell measurement on
$A_cU_{L_c}$ has been applied, we now apply
the Bell measurement on 
$A_{\bar c}U_{L_{\bar c}}$.  Denote this
trace-preserving measurement map by
$
\mathcal N_{A_{\bar c}U_{L_{\bar c}}\to Z},
$
where $Z$ contains the Bell-measurement outcome $T_{\bar c}$ and the
post-measurement register. Since the correctness condition only keeps
the marginal on $B_cR$, the registers $Z W_{L_{\bar c}}$ are traced
out. Therefore, by
\eqref{eq:correctness_selected_after_correction} and the fact that
$\mathcal N_{A_{\bar c}U_{L_{\bar c}}\to Z}$ is trace-preserving, we
have
\begin{align}
&\operatorname{Tr}_{ZW_{L_{\bar c}}}
\left[
\left(
\operatorname{id}_{B_cRW_{L_{\bar c}}}
\otimes
\mathcal N_{A_{\bar c}U_{L_{\bar c}}\to Z}
\right)
\left(
\left(
\operatorname{id}_{A_c\to B_c}
\otimes
\operatorname{id}_{A_{\bar c}R}
\right)
\left(\Psi^{A_cA_{\bar c}R}\right)
\otimes
\Phi_{2^m}^{U_{L_{\bar c}}W_{L_{\bar c}}}
\right)
\right]
\nonumber\\
&\qquad=
\operatorname{Tr}_{A_{\bar c}U_{L_{\bar c}}W_{L_{\bar c}}}
\left[
\left(
\operatorname{id}_{A_c\to B_c}
\otimes
\operatorname{id}_{A_{\bar c}R}
\right)
\left(\Psi^{A_cA_{\bar c}R}\right)
\otimes
\Phi_{2^m}^{U_{L_{\bar c}}W_{L_{\bar c}}}
\right]
\nonumber\\
&\qquad=
\left(
\operatorname{id}_{A_c\to B_c}
\otimes
\operatorname{id}_R
\right)
\left(\Psi^{A_cR}\right).
\label{eq:correctness_after_unselected_tracing}
\end{align}

The expression in \eqref{eq:correctness_after_unselected_tracing} does
not depend on $t_c$, $\mathcal G$, $L_c$, or
$L_{\bar c}$, hence, averaging over these variables  gives \eqref{eq:qot_correctness}.

\subsubsection{Receiver privacy}

The following symmetry property shows that Bob's subset announcement is
independent of his choice $C$.

\begin{lemma}
\label{lem:subset_announcement_symmetry}
Let $\mathcal J_1,\mathcal J_2\subseteq[n]$. Then,
\begin{equation}
\Prb[
L_1=\mathcal J_1,L_2=\mathcal J_2,F=0
\mid C=1
]
=
\Prb[
L_1=\mathcal J_1,L_2=\mathcal J_2,F=0
\mid C=2
]. \label{eqrps}
\end{equation}
\end{lemma}
 
\begin{proof}
Define
$
\mathfrak G
\triangleq
\{
\mathcal T\subseteq[n]:
|\mathcal T|\geq m,\ |\mathcal T^c|\geq m
\}.
$
If $|\mathcal J_1|\neq m$, $|\mathcal J_2|\neq m$, or
$\mathcal J_1\cap\mathcal J_2\neq\varnothing$, both sides of \eqref{eqrps} are zero. Hence, suppose that
$
|\mathcal J_1|=|\mathcal J_2|=m
$
and
$
\mathcal J_1\cap\mathcal J_2=\varnothing.
$
Then, we have
\begin{align}
\Prb[
L_1=\mathcal J_1,L_2=\mathcal J_2,F=0
\mid C=1
]
\nonumber
&=
\sum_{\mathcal T\subseteq[n]}
\Prb[\mathcal G=\mathcal T\mid C=1]\,
\Prb[
L_1=\mathcal J_1,L_2=\mathcal J_2,F=0
\mid \mathcal G=\mathcal T,C=1
]
\nonumber\\
&\overset{(a)}{=}
\sum_{\mathcal T\subseteq[n]}
\Prb[\mathcal G=\mathcal T]\,
\Prb[
L_1=\mathcal J_1,L_2=\mathcal J_2,F=0
\mid \mathcal G=\mathcal T,C=1
]
\nonumber\\
&\overset{(b)}{=}
\sum_{\mathcal T\in\mathfrak G}
\Prb[\mathcal G=\mathcal T]\,
\Prb[
L_1=\mathcal J_1,L_2=\mathcal J_2
\mid \mathcal G=\mathcal T,C=1
]
\nonumber\\
&\overset{(c)}{=}
\sum_{\mathcal T\in\mathfrak G}
\Prb[\mathcal G=\mathcal T]\,
\frac{
\mathds{1}\{
\mathcal J_1\subseteq\mathcal T,\
\mathcal J_2\subseteq\mathcal T^c
\}
}{
\binom{|\mathcal T|}{m}
\binom{n-|\mathcal T|}{m}
}
\nonumber\\
&\overset{(d)}{=}
\sum_{\mathcal T\in\mathfrak G}
\Prb[\mathcal G=\pi(\mathcal T)]\,
\frac{
\mathds{1}\{
\mathcal J_2\subseteq\pi(\mathcal T),\
\mathcal J_1\subseteq\pi(\mathcal T)^c
\}
}{
\binom{|\pi(\mathcal T)|}{m}
\binom{n-|\pi(\mathcal T)|}{m}
}
\nonumber\\
&\overset{(e)}{=}
\sum_{\mathcal T'\in\mathfrak G}
\Prb[\mathcal G=\mathcal T']\,
\frac{
\mathds{1}\{
\mathcal J_2\subseteq\mathcal T',\
\mathcal J_1\subseteq(\mathcal T')^c
\}
}{
\binom{|\mathcal T'|}{m}
\binom{n-|\mathcal T'|}{m}
}
\nonumber\\
&\overset{(f)}{=}
\Prb[
L_1=\mathcal J_1,L_2=\mathcal J_2,F=0
\mid C=2
],
\nonumber
\end{align}
where
\begin{enumerate}[(a)]
\item holds because $\mathcal G$ is independent of $C$;
\item holds because $F=0$ if and only if
$\mathcal G\in\mathfrak G$;
\item holds by Line 3 of
Algorithm~\ref{alg:qot_achievability};
\item holds with $\pi$ a permutation of $[n]$ that swaps
$\mathcal J_1$ and $\mathcal J_2$ and fixes all elements outside
$\mathcal J_1\cup\mathcal J_2$, because the erasure pattern is i.i.d.,
$|\pi(\mathcal T)|=|\mathcal T|$, so $\Prb[\mathcal G=\mathcal T]=\Prb[\mathcal G=\pi(\mathcal T)]$, and
$
(\mathcal J_1\subseteq\mathcal T,\
\mathcal J_2\subseteq\mathcal T^c)
\Longleftrightarrow
(\mathcal J_2\subseteq\pi(\mathcal T),\
\mathcal J_1\subseteq\pi(\mathcal T)^c);
$
\item follows from the change of variables
$\mathcal T'=\pi(\mathcal T)$;
\item follows as in $(c)$ with $C=2$.
\end{enumerate}
\end{proof}

We first prove \eqref{eq:qot_receiver_privacy1}. When $F=0$,
for   
$m_0=(\mathcal J_1,\mathcal J_2,t_1,t_2)$, we have
\begin{align*}
&\Prb[M_0=m_0,F=0\mid C=c]\\
&=
\Prb[L_1=\mathcal J_1,L_2=\mathcal J_2,F=0\mid C=c] 
\Prb[T_1=t_1,T_2=t_2
\mid L_1=\mathcal J_1,L_2=\mathcal J_2,F=0,C=c]\\
&\overset{(a)}{=}
4^{-2m}
\Prb[L_1=\mathcal J_1,L_2=\mathcal J_2,F=0\mid C=c]\\
& \overset{(b)}{=} \Prb[M_0=m_0,F=0\mid C=\bar{c}],
\end{align*}
where \begin{enumerate}[(a)]
    \item   holds by Lemma~\ref{lem:teleportation_properties_ach};
 \item holds by
Lemma~\ref{lem:subset_announcement_symmetry}.
\end{enumerate}
 When $F=1$, $(M_0,F)=(\bot,1)$, whose probability is  independent
of $c$, since $\mathcal G$ is independent of $c$ and
determines the abort decision. Hence, \eqref{eq:qot_receiver_privacy1} holds.

We next prove \eqref{eq:qot_receiver_privacy2}.
Let $\rho$ denote the state immediately after Line~1
of Algorithm~\ref{alg:qot_achievability}.
For each $(m_0,f)$, let
$\mathcal A_{m_0,f}:\mathcal L(A_1A_2U^n)
\to\mathcal L(Q_A)$ denote Alice's completely positive,
trace-nonincreasing operation corresponding to
the Bell measurements in Line~4, conditioned on $m_0$, when $f=0$,
or to her termination after Bob's abort announcement
in Line~2 when $f=1$.
For fixed $(m_0,f)$, this operation has no further
dependence on $c$ or $s$.

When $F=0$, for
$m_0=(\mathcal J_1,\mathcal J_2,t_1,t_2)$
and every compatible $s$, we have
\begin{align*}
\omega^{(c)}_{Q_AR
\mid M_0=m_0,F=0,\mathsf S=s}(\Psi)
&\overset{(a)}{=}
4^{2m}
(\mathcal A_{m_0,0}\otimes\operatorname{id}_R)
\left(
\rho^{A_1A_2RU^n
\mid\mathcal G=s^c,C=c}(\Psi)
\right)\\
&\overset{(b)}{=}
4^{2m}
(\mathcal A_{m_0,0}\otimes\operatorname{id}_R)
\left(
\Psi^{A_1A_2R}\otimes\frac{I_{U^n}}{2^n}
\right)\\
&\triangleq\tau_{Q_ARm_00}(\Psi),
\end{align*}
where \begin{enumerate}[(a)]
\item holds because Bob's announcements in Lines~2--3
depend only on $\mathcal G$ and independent randomness,
so conditioning on them leaves the state on
$A_1A_2RU^n$ unchanged given $\mathcal G,C$, then
conditioning on the Bell outcomes in Line~4 divides
the subnormalized output by their probability
$4^{-2m}$, given by
Lemma~\ref{lem:teleportation_properties_ach}, finally,
 Bob's decoding in Line~5 acts only on his registers,
leaving $Q_AR$ unchanged.

\item holds because Line~1 prepares $n$ independent
Bell pairs independently of $\Psi$, and conditioning
on $\mathcal G$ only specifies whether each receiving
half is held by Bob or the environment, leaving
Alice's marginal equal to $I_{U^n}/2^n$.
\end{enumerate}
Thus, for fixed $m_0$, both the input state
$\Psi^{A_1A_2R}\otimes I_{U^n}/2^n$ and the operation
$\mathcal A_{m_0,0}$ are independent of $c$ and $s$,
so $\tau_{Q_ARm_00}(\Psi)$ is also independent
of $c$ and $s$.
When $F=1$, the protocol terminates in Line~2,
before the Bell measurements in Line~4, and $m_0=\bot$.
Alice applies the fixed abort operation $\mathcal A_{\bot,1}$
to her registers $A_1A_2U^n$, whose joint state with $R$
after Line~1 is $\Psi^{A_1A_2R}\otimes I_{U^n}/2^n$
regardless of the unerased set $\mathcal G$.
Since both this state and the abort operation are
independent of $c$ and $s$, the resulting state on
$Q_AR$, denoted by $\tau_{Q_AR\bot1}(\Psi)$,
is also independent of $c$ and $s$.
Hence, \eqref{eq:qot_receiver_privacy2} holds.

\subsubsection{Non-chosen state privacy}
To prove non-chosen-state privacy, we first show that,
on success, the conditional state on $B_cE_cR$ depends
only on $\Psi^{A_cR}$, and that the erasure pattern,
public communication, and abort flag have a joint
distribution independent of $\Psi$.
\begin{lemma}
\label{lem:honest_nonchosen_view}
Fix $c\in\{1,2\}$.
Conditioned
on $F=0$ and any event of positive probability
$(\mathcal G,L_1,L_2,T_1,T_2)
=(\mathcal T,\mathcal J_1,\mathcal J_2,t_1,t_2)$,
the state on $B_cE_cR$ has the form
$\Psi^{B_cR}\otimes
\xi_{\mathcal T,\mathcal J_1,\mathcal J_2,t_1,t_2}^{E_c}$,
where $\xi_{\mathcal T,\mathcal J_1,\mathcal J_2,t_1,t_2}^{E_c}$
is normalized and independent of $\Psi$.
Moreover, with $M_0=(L_1,L_2,T_1,T_2)$ on success,
where $(L_1,L_2)$ are announced in Line~3 and
$(T_1,T_2)$ in Line~4, and  $M_0=\bot$ on abort,
the joint distribution of $(\mathcal G,M_0,F)$
is independent of $\Psi$.
\end{lemma}

\begin{proof}
We have
\begin{align}
&\omega^{(c)}_{B_cE_cR
\mid
\mathcal G=\mathcal T,\,
L_1=\mathcal J_1,\,
L_2=\mathcal J_2,\,
T_1=t_1,\,
T_2=t_2,\,
F=0}(\Psi)
\nonumber\\
&\overset{(a)}{=}
\operatorname{Tr}_{W_{\mathcal J_{\bar c}}}
\left[
P_{t_{\bar c}}^{W_{\mathcal J_{\bar c}}}
\left[
\left(
\operatorname{id}_{A_c\to B_c}
\otimes
\operatorname{id}_{A_{\bar c}\to W_{\mathcal J_{\bar c}}}
\otimes
\operatorname{id}_R
\right)
\left(
\Psi^{A_cA_{\bar c}R}
\right)
\right]
\left(
P_{t_{\bar c}}^{W_{\mathcal J_{\bar c}}}
\right)^\dagger
\right]
\otimes
\xi_{\mathcal T,\mathcal J_1,\mathcal J_2,t_1,t_2}^{E_c}
\nonumber\\
&\overset{(b)}{=}
\left(
\operatorname{id}_{A_c\to B_c}
\otimes
\operatorname{id}_R
\right)
\left(
\Psi^{A_cR}
\right)
\otimes
\xi_{\mathcal T,\mathcal J_1,\mathcal J_2,t_1,t_2}^{E_c}
\nonumber\\
&=
\Psi^{B_cR}
\otimes
\xi_{\mathcal T,\mathcal J_1,\mathcal J_2,t_1,t_2}^{E_c},
\nonumber
\end{align}
where
\begin{enumerate}[(a)]
\item holds because, immediately before Line~4,
conditioned on $\mathcal G=\mathcal T$ and
$(L_1,L_2)=(\mathcal J_1,\mathcal J_2)$,
we have
$\rho^{U_{\mathcal J_c}V_{\mathcal J_c}}
=\Phi_{2^m}^{U_{\mathcal J_c}V_{\mathcal J_c}}$
and
$\rho^{U_{\mathcal J_{\bar c}}W_{\mathcal J_{\bar c}}}
=\Phi_{2^m}^{U_{\mathcal J_{\bar c}}W_{\mathcal J_{\bar c}}}$.
Applying Lemma~\ref{lem:teleportation_properties_ach}
to the two teleportations in Line~4 transfers
$A_c$ to Bob and $A_{\bar c}$ to the environment.
Bob's correction in Line~5 cancels the Pauli
associated with $t_c$.
His remaining systems, including unused channel
outputs and retained classical records, have a
normalized conditional state
$\xi_{\mathcal T,\mathcal J_1,\mathcal J_2,t_1,t_2}^{E_c}$
independent of $\Psi$;

\item holds because the partial trace over
$W_{\mathcal J_{\bar c}}$ is invariant under
$P_{t_{\bar c}}^{W_{\mathcal J_{\bar c}}}$
and corresponds to tracing $A_{\bar c}$ from
the input state.
\end{enumerate}

The same applications of
Lemma~\ref{lem:teleportation_properties_ach} give
$\Prb[T_1=t_1,T_2=t_2
\mid \mathcal G=\mathcal T,L_1=\mathcal J_1,
L_2=\mathcal J_2,F=0,C=c]=4^{-2m}$
for every input state.
Since Lines~1--3 generate the erasure pattern,
abort decision, and subset choices without using
either input, the probabilities of all successful
records $(\mathcal G,M_0,F)$ are input-independent.
On abort, $M_0=\bot$, and the decision in Line~2
depends only on $\mathcal G$.
Thus, the probabilities of the aborted records
are also input-independent.
\end{proof}

We now use Lemma~\ref{lem:honest_nonchosen_view}
to construct the simulator.
For every state $\rho^{A_c}$, define
$\mathsf{Sim}_{c,n}(\rho^{A_c})
\triangleq
\omega^{(c)}_{B_cE_cM_0F}
\left(
\rho^{A_c}\otimes
\ket{0^m}\!\bra{0^m}^{A_{\bar c}}
\right)$,
which appends the fixed non-chosen input,
executes both parties and the erasure channel in
Algorithm~\ref{alg:qot_achievability} with choice $c$,
and returns $B_cE_cM_0F$.
This map is CPTP because appending a fixed state, executing the protocol, and tracing out the other output registers
are CPTP operations, which are independent of $\Psi$. Then, \eqref{eq:qot_global_destructive_privacy_general}
holds because
\begin{align}
\omega^{(c)}_{B_cE_cM_0FR}
\left(\Psi^{A_1A_2R}\right)
&\overset{(a)}{=}
\omega^{(c)}_{B_cE_cM_0FR}
\left(
\Psi^{A_cR}\otimes
\ket{0^m}\!\bra{0^m}^{A_{\bar c}}
\right)
\nonumber\\
&\overset{(b)}{=}
(\mathsf{Sim}_{c,n}\otimes\operatorname{id}_R)
\left(\Psi^{A_cR}\right),
\nonumber
\end{align}
where
\begin{enumerate}[(a)]
\item holds by comparing Algorithm~\ref{alg:qot_achievability},
with choice $c$, on inputs $\Psi^{A_1A_2R}$ and
$\Psi^{A_cR}\otimes
\ket{0^m}\!\bra{0^m}^{A_{\bar c}}$. First, consider $F=0$ and fix
$\mathcal G=\mathcal T$ and
$M_0=m_0=(\mathcal J_1,\mathcal J_2,t_1,t_2)$.
Lemma~\ref{lem:honest_nonchosen_view} gives
\begin{align*}
\omega^{(c)}_{B_cE_cR
\mid \mathcal G=\mathcal T,M_0=m_0,F=0}
\left(\Psi^{A_1A_2R}\right)
=
\omega^{(c)}_{B_cE_cR
\mid \mathcal G=\mathcal T,M_0=m_0,F=0}
\left(
\Psi^{A_cR}\otimes
\ket{0^m}\!\bra{0^m}^{A_{\bar c}}
\right). \numberthis \label{eq:ppp}
\end{align*}
Lemma~\ref{lem:honest_nonchosen_view}  also gives the same probabilities for
$(\mathcal G,M_0)$ conditioned on $F=0$ so  averaging \eqref{eq:ppp}  gives
identical states on $B_cE_cM_0R$ conditioned on
success. When $F=1$, Line~2 terminates the protocol before
either input is used.
Let $\xi_{\bot,c}^{B_cE_cM_0}$ denote Bob's resulting
state, which is independent of the input and
uncorrelated with $R$.
For each of the two inputs, the conditional state
on $B_cE_cM_0R$ is therefore  equal to
$\xi_{\bot,c}^{B_cE_cM_0}\otimes\Psi^R$. Finally, since Lemma~\ref{lem:honest_nonchosen_view} also
gives the same distribution of $F$ for both inputs,
averaging over $F$ proves $(a)$;

\item holds by the definition of $\mathsf{Sim}_{c,n}$.
\end{enumerate}

\section{Converse proof of Theorem \ref{thm:finite_length_achievability}}
\label{sec:converse}
Fix an arbitrary $(n,m)$ protocol satisfying Conditions 1--4 of Definition~\ref{def:qot_delta_achievability} with abort probability at most $\delta$. This includes the honest execution of every admissible malicious protocol, so the proof applies to both settings.   For the converse, we use the maximally retained Stinespring
implementation described in Definition \ref{def1}, hence,   conditioned on the complete public transcript, the abort flag,
and the erasure pattern, the resulting global state is pure.

In the following, let
$Z^n=(Z_1,\ldots,Z_n)\in\{0,1\}^n$ with
$
Z_i\triangleq \mathds{1}\{\text{the $i$-th channel output is not erased at Bob}\}.
$
Define
$$
G(Z^n)\triangleq\{i:Z_i=1\},
\qquad
S(Z^n)\triangleq\{i:Z_i=0\}.
$$
 First, we show  in
Section~\ref{secb1} that exact
recovery of the chosen $m$-qubit state requires at least $m$ unerased
channel uses on every successful execution, i.e., $|G(Z^n)|\geq m$. Then, we show in Section \ref{secb2} that privacy
of the non-chosen $m$-qubit state requires at least $m$ erased channel
uses on every successful execution, i.e., $|S(Z^n)|\geq m$. Finally, we conclude in Section \ref{secb3} that $
\Prb\left[
m\leq \operatorname{Bin}(n,\epsilon)\leq n-m
\right]
\geq
1-\delta$. 

\subsection{Unerased-channel-use bound}
\label{secb1}

Fix $C=1$ and choose the input as
$
\Psi^{A_1A_2R_1R_2}
=
\Phi^{A_1R_1}_{2^m}\otimes \xi^{A_2R_2},
$
where $\xi^{A_2R_2}$ is a pure state and $\Phi^{A_1R_1}_{2^m}$ is maximally entangled, i.e., 
$
\Phi^{A_1R_1}_{2^m}
=
|\Phi_{2^m}\rangle\!\langle\Phi_{2^m}|^{A_1R_1}
$
and
$
|\Phi_{2^m}\rangle^{A_1R_1}
=
\frac{1}{\sqrt{2^m}}
\sum_{x\in\{0,1\}^m}
|x\rangle^{A_1}|x\rangle^{R_1}.
$
All probabilities and states in this subsection are evaluated for this fixed
choice of $C$ and $\Psi$.

For a bipartite pure state $\psi^{AB}$, denote its Schmidt rank with
respect to the bipartition $A:B$ by
$\operatorname{SR}_{A:B}(\psi)$. Conditioned on a fixed erasure pattern
$Z^n=z^n$, only the unerased quantum outputs indexed by $G(z^n)$ enter
Bob's private registers, while the erased quantum outputs are routed to the
environment. The following lemma formalizes that each unerased qubit can
increase the Schmidt rank across Bob versus the rest by at most a factor of
$2$.

\begin{lemma}
\label{lem:schmidt_rank_growth}
Fix   $z^n$,   $m_0$ , and  $f$   such that
$
\Prb[Z^n=z^n,M_0=m_0,F=f]>0.
$
For each $i\in\{0,\ldots,n\}$, let $\psi_i$ be the pure state immediately
after the first $i$ channel uses, conditioned only on $Z^i=z^i$ and on
the public messages produced by that time taking their values in $m_0$.
Let $\mathcal L_i$ be the collection of registers outside Bob's
laboratory, including Alice's private register, the reference systems
$R_1R_2$, and the erasure-channel environment registers $W^i$, and let
$\mathcal B_i$ be Bob's private registers. Define
$
R_i\triangleq
\operatorname{SR}_{\mathcal L_i:\mathcal B_i}(\psi_i).
$
Then,
\begin{align}
R_0&=1,
\label{eq:sr_initial_one}\\
R_i&\leq 2^{z_i}R_{i-1},
\qquad
i\in\{1,\ldots,n\}.
\label{eq:sr_one_step_growth}
\end{align}
Hence,
$
R_n\leq 2^{|G(z^n)|}.
$

\end{lemma}

\begin{proof}See Appendix \ref{app_B}.\end{proof}

The next lemma converts Lemma~\ref{lem:schmidt_rank_growth} into a
statement about aborts. Specifically,  if fewer than $m$ channel uses are not erased,
the protocol must abort.

\begin{lemma}
\label{lem:bad_erasure_no_success}
For every $z^n$ such that
$
|G(z^n)|<m,
$
we have
$
\Prb[F=0,Z^n=z^n]=0.
$
\end{lemma}

\begin{proof}
Fix $z^n$ such that
$
 |G(z^n)|<m
$
and, by contradiction, suppose that
$
\Prb[F=0,Z^n=z^n]>0.
$
Since
 $
\Prb[F=0,Z^n=z^n]
=
\sum_{m_0}
\Prb[F=0,M_0=m_0,Z^n=z^n],
$
there exists $m_0$ such that
$
\Prb[F=0,M_0=m_0,Z^n=z^n]>0,
$
so that the normalized conditional state
$\Omega^{(1,z^n,m_0)}_{Q_AB_1E_1R_1R_2W^n}
\triangleq
\Omega^{(1)}_{Q_AB_1E_1R_1R_2W^n
\,\mid\,
F=0,M_0=m_0,Z^n=z^n}(\Psi)
$
is well defined. Then, we have
\begin{align}
\operatorname{SR}_{\,Q_AR_1R_2W^n\,:\,B_1E_1}
\left(
\Omega^{(1,z^n,m_0)}
\right)
&\overset{(a)}{=}
\operatorname{SR}_{R_1:B_1}(\Phi_{2^m})\,
\operatorname{SR}_{\,Q_AR_2W^n\,:\,E_1}(\chi)
\nonumber\\
&\overset{(b)}{\geq}
\operatorname{SR}_{R_1:B_1}(\Phi_{2^m})
\nonumber\\
&\overset{(c)}{=}
2^m,
\label{eq:schmidt_lower_block}
\end{align}
where 
\begin{enumerate}[(a)]
    \item  follows because there exists a pure state $\chi^{Q_AE_1R_2W^n}$ such that $\Omega^{(1,z^n,m_0)}_{Q_AB_1E_1R_1R_2W^n}=\Phi^{B_1R_1}_{2^m}\otimes \chi^{Q_AE_1R_2W^n}$ since $\Omega^{(1,z^n,m_0)}$ is pure  and its $B_1R_1$-marginal $\omega^{(1)}_{B_1R_1\,\mid\, F=0,M_0=m_0,Z^n=z^n}(\Psi)$ is the pure state $\Phi^{B_1R_1}_{2^m}$. Indeed, we have $\omega^{(1)}_{B_1R_1\,\mid\, F=0,M_0=m_0,Z^n=z^n}(\Psi)=\Phi^{B_1R_1}_{2^m}$   by the fact that every positive-weight term in a convex
decomposition of a pure state equals that pure state
\cite[Sec.~2.1.2]{watrous2018theory}, since $\Phi^{B_1R_1}_{2^m}$ is pure, $\Prb[M_0=m_0,Z^n=z^n\mid F=0]>0$,  and
\begin{align}
\Phi^{B_1R_1}_{2^m}
=
\omega^{(1)}_{B_1R_1\mid F=0}(\Psi) \nonumber&=
\sum_{\tilde m_0,\tilde z^n}
\Prb[M_0=\tilde m_0,Z^n=\tilde z^n\mid F=0]\,
\omega^{(1)}_{B_1R_1
\,\mid\,
F=0,M_0=\tilde m_0,Z^n=\tilde z^n}(\Psi),
\end{align}
where  the first equality  follows from  \eqref{eq:qot_correctness}, the second equality   follows from the law of
total probability over the classical registers $M_0$ and $Z^n$, conditioned
on $F=0$;
\item follows because $ \operatorname{SR}_{\,Q_AR_2W^n\,:\,E_1}(\chi)\geq1 $;
\item follows because $\Phi_{2^m}$ is maximally entangled and thus has Schmidt rank $2^m$.
\end{enumerate}

On the other hand, Lemma~\ref{lem:schmidt_rank_growth} gives
\begin{align*}
\operatorname{SR}_{\,Q_AR_1R_2W^n\,:\,B_1E_1}
\left(
\Omega^{(1,z^n,m_0)}
\right)
& \leq R_n \\
&\leq
2^{|G(z^n)|}\\
&<
2^m,
\end{align*}
which contradicts \eqref{eq:schmidt_lower_block}. Therefore,
$
\Prb[F=0,Z^n=z^n]=0.
$
\end{proof}

\subsection{Erased-channel-use bound}

\label{secb2}
Choose the maximally entangled input $
\Psi^{A_1A_2R_1R_2}
=
\Phi^{A_1R_1}_{2^m}
\otimes
\Phi^{A_2R_2}_{2^m}$. All probabilities and states in this subsection are evaluated for this
input.

\begin{lemma}
\label{lem:bad_privacy_erasure_no_success}
For every $z^n$ such that
$
|S(z^n)|<m,
$
we have
$
\Prb[F=0,Z^n=z^n\mid C=1]=0.
$
\end{lemma}

\begin{proof}
Fix $z^n$ such that
$
|S(z^n)|<m
$
and suppose, by contradiction, that
$
\Prb[F=0,Z^n=z^n\mid C=1]>0$ so that for 
$
s\triangleq S(z^n)
$, there exists $m_0$ such that
\begin{equation}
\Prb[
F=0,M_0=m_0,\mathsf S=s
\mid
C=1
]
>0.
\label{eq:bad_privacy_transcript_positive}
\end{equation}
By \eqref{eq:qot_receiver_privacy1} and \eqref{eq:bad_privacy_transcript_positive},
$
\Prb[
F=0,M_0=m_0
\mid
C=2
]
=
\Prb[
F=0,M_0=m_0
\mid
C=1
]>0$ so that there exists an erased set $s'$ such that
$
\Prb[
F=0,M_0=m_0,\mathsf S=s'
\mid
C=2
]
>0. $
Then, define the normalized conditional  states
\begin{align*}
\sigma_{m_0,s}^{(1)}
&\triangleq
\Omega^{(1)}_{Q_AB_1E_1R_1R_2W^n
\mid
F=0,M_0=m_0,\mathsf S=s}(\Psi),\\
\sigma_{m_0,s'}^{(2)}
&\triangleq
\Omega^{(2)}_{Q_AB_2E_2R_1R_2W^n
\mid
F=0,M_0=m_0,\mathsf S=s'}(\Psi).
\end{align*}

Define $X\triangleq Q_AR_1R_2$, $P_1\triangleq B_1E_1W^n$, $P_2\triangleq B_2E_2W^n$, and by  receiver's choice privacy
\eqref{eq:qot_receiver_privacy2},
\begin{equation}
\left(\sigma_{m_0,s}^{(1)}\right)_X
=
\left(\sigma_{m_0,s'}^{(2)}\right)_X.
\label{eq:branchwise_alice_state_equality}
\end{equation}

Under the maximally retained implementation, $(\sigma_{m_0,s}^{(1)})_{XP_1}$ and
$(\sigma_{m_0,s'}^{(2)})_{XP_2}$ are pure and  two
purifications of the same state on $X$ by \eqref{eq:branchwise_alice_state_equality}. Hence, by \cite[Th. 5.1.1]{wilde2013quantum}, there exists an isometry
$
U:P_1\to P_2K
$
such that, with
$
\widehat{\sigma}_{XP_2K}
\triangleq
(I_X\otimes U)
\left(\sigma_{m_0,s}^{(1)}\right)_{XP_1}
(I_X\otimes U^\dagger),
$
we have
$
\operatorname{Tr}_K[\widehat{\sigma}_{XP_2K}]
=
\left(\sigma_{m_0,s'}^{(2)}\right)_{XP_2}.$ Then,  we have
\begin{align*}
2m
&\overset{(a)}{=}
I(R_2;B_2)_{\sigma_{m_0,s'}^{(2)}}\\
&\overset{(b)}{=}
I(R_2;B_2)_{\widehat{\sigma} }\\
&\overset{(c)}{\leq}
I(R_2;P_2K)_{\widehat{\sigma} }\\
&\overset{(d)}{=}
I(R_2;P_1)_{\sigma_{m_0,s}^{(1)}}\\
&=
I(R_2;B_1E_1W^n)_{\sigma_{m_0,s}^{(1)}} \\
&\overset{(e)}{\leq}
2\log\dim R_2\\
&=
2m, \numberthis\label{eq:environment_information_equality}
\end{align*}
where
\begin{enumerate}[(a)]
    \item holds because by correctness \eqref{eq:qot_correctness} for $C=2$, 
$
\Phi_{2^m}^{R_2B_2}
=
\sum_{\widetilde m_0,\widetilde s}
\Prb[M_0=\widetilde m_0,\mathsf S=\widetilde s
\mid F=0,C=2]\,
\left(\sigma_{\widetilde m_0,\widetilde s}^{(2)}\right)_{R_2B_2},
$
and since $\Phi_{2^m}^{R_2B_2}$ is pure, every positive-weight
term in this mixture equals $\Phi_{2^m}^{R_2B_2}$
\cite[Sec.~2.1.2]{watrous2018theory}, in particular, $
\left(\sigma_{m_0,s'}^{(2)}\right)_{R_2B_2}
=
\Phi_{2^m}^{R_2B_2}
$ so that 
    $
    I(R_2;B_2)_{\sigma_{m_0,s'}^{(2)}}
    =
    H(R_2)+H(B_2)-H(R_2B_2)
    =
    2m
    $;
    \item holds because
    $
    \operatorname{Tr}_K\widehat{\sigma}_{ XP_2K}
    =
    (\sigma_{m_0,s'}^{(2)})_{XP_2}
    $,
    in particular,
    $
    \widehat{\sigma}_{R_2B_2}
    =
    (\sigma_{m_0,s'}^{(2)})_{R_2B_2}
    $;
    \item holds by the data processing inequality
    \cite[Theorem~11.9.4]{wilde2013quantum};
    \item holds by the invariance of quantum entropy under the isometry
    $U:P_1\to P_2K$
    \cite[Prop.~11.1.5]{wilde2013quantum}, since
    \begin{align*}
    \widehat{\sigma}_{R_2P_2K}
    =
    (I_{R_2}\otimes U)
    (\sigma_{m_0,s}^{(1)})_{R_2P_1}
    (I_{R_2}\otimes U^\dagger);
    \end{align*}
    \item follows from the mutual information dimension bound
    \cite[Eq. (11.73)]{wilde2013quantum}.
\end{enumerate}
Then, we have
\begin{align*}
I(R_2;W^n\mid B_1E_1)_{\sigma_{m_0,s}^{(1)}}
& = I(R_2;B_1E_1W^n)_{\sigma_{m_0,s}^{(1)}} - I(R_2;B_1E_1)_{\sigma_{m_0,s}^{(1)}}\\
&\overset{(a)}{=}
2m - I(R_2;B_1E_1)_{\sigma_{m_0,s}^{(1)}}\\
& \overset{(b)}{=} 2m, \numberthis
\label{eq:branchwise_environment_conditional_information}
\end{align*}
where
\begin{enumerate}[(a)]
    \item holds by \eqref{eq:environment_information_equality};
    \item holds because by non-chosen-state privacy \eqref{eq:qot_global_destructive_privacy_general},   
\begin{align*}
\omega^{(1)}_{B_1E_1M_0FR_1R_2}(\Psi)
&=
\left(
\mathsf{Sim}_{1,n}\otimes\operatorname{id}_{R_1R_2}
\right)
\left(
\Psi^{A_1R_1R_2}
\right)
\\
&=
\left(
\mathsf{Sim}_{1,n}\otimes\operatorname{id}_{R_1R_2}
\right)
\left(
\operatorname{Tr}_{A_2}\!\left[
\Phi_{2^m}^{A_1R_1}
\otimes
\Phi_{2^m}^{A_2R_2}
\right]
\right)
\\
&=
\left(
\mathsf{Sim}_{1,n}\otimes\operatorname{id}_{R_1R_2}
\right)
\left(
\Phi_{2^m}^{A_1R_1}
\otimes
\frac{I_{R_2}}{2^m}
\right)
\\
&=
\left[
\left(
\mathsf{Sim}_{1,n}\otimes\operatorname{id}_{R_1}
\right)
\left(
\Phi_{2^m}^{A_1R_1}
\right)
\right]
\otimes
\frac{I_{R_2}}{2^m}
\\
&=
\omega^{(1)}_{B_1E_1M_0FR_1}(\Psi)
\otimes
\frac{I_{R_2}}{2^m},
\end{align*}
and since $\mathsf S$ is a classical subregister of $E_1$, conditioning
this product state on $ F = 0$, $M_0=m_0$ and $\mathsf S=s$ preserves the
product form so that
$\left(\sigma_{m_0,s}^{(1)}\right)_{R_2B_1E_1}
=
\frac{I_{R_2}}{2^m}
\otimes
\left(\sigma_{m_0,s}^{(1)}\right)_{B_1E_1},$ 
and $
I(R_2;B_1E_1)_{\sigma_{m_0,s}^{(1)}}=0.$
\end{enumerate}

On the other hand,
\begin{align*}
I(R_2;W^n\mid B_1E_1)_{\sigma_{m_0,s}^{(1)}}
&\leq
2H(W^n)_{\sigma_{m_0,s}^{(1)}}
\\
&\leq
2\log
\dim\operatorname{supp}
\left[
\left(\sigma_{m_0,s}^{(1)}\right)_{W^n}
\right]\\
&\leq
2|s|, \numberthis \label{eq:branchwise_environment_conditional_informationb}
\end{align*}
where the first inequality holds by \cite[Eq. (11.102)]{wilde2013quantum}, the second inequality holds by \cite[Prop.~11.1.3]{wilde2013quantum}, the last inequality holds because, conditioned on $\mathsf S=s$, the environment systems at
unerased positions are fixed erasure flags, while the systems at erased
positions are qubits. 

Hence, 
\eqref{eq:branchwise_environment_conditional_information} and \eqref{eq:branchwise_environment_conditional_informationb}  give the contradiction $2m\leq 2|s|<2m$. Therefore,
 $
\Prb[F=0,Z^n=z^n\mid C=1]=0.
 $
\end{proof}
\subsection{Combining the unerased-channel-use and erased-channel-use bounds}
\label{secb3}

Choose the input
$\Psi^{A_1A_2R_1R_2}
=\Phi_{2^m}^{A_1R_1}\otimes\Phi_{2^m}^{A_2R_2}$
and fix $C=1$. Then,
\begin{align*}
1-\delta
&\overset{(a)}{\leq}
\Prb[F=0\mid C=1]\\
&\overset{(b)}{\leq}
\Prb[m\leq |S(Z^n)|\leq n-m\mid C=1]\\
&\overset{(c)}{=}
\Prb[m\leq\operatorname{Bin}(n,\epsilon)\leq n-m],
\end{align*}
where \begin{enumerate}[(a)]
    \item 
 follows from \eqref{eq:qot_abort}; \item follows from
Lemmas~\ref{lem:bad_erasure_no_success} and
\ref{lem:bad_privacy_erasure_no_success}, which apply to
this same input and require at least $m$ unerased and
$m$ erased channel uses upon success; \item
 holds because the number of erased channel uses
has distribution $\operatorname{Bin}(n,\epsilon)$
independently of $C$.
\end{enumerate}

\section{Proof of Theorem \ref{thm:malicious_finite}}
\label{sec:malicious_ach}

 We first recall interactive
hashing and  a property of Clifford unitaries in
Section~\ref{subsec:malicious_prelim}. We then present the protocol
and its analysis in
Sections~\ref{subsec:malicious_scheme} and
\ref{subsec:malicious_analysis}, respectively. 

\subsection{Preliminaries}
\label{subsec:malicious_prelim}

\subsubsection{Interactive hashing}

Interactive hashing takes $W\in\{0,1\}^h$ from Bob and no input
from Alice, and produces two public strings $W_0,W_1$.
We use the protocol guaranteed by
\cite[Lemma~II.5]{konig2012unconditional}, with the following properties:
\begin{enumerate}
\item If Bob is honest, $W_0\neq W_1$ and $W=W_D$ for some
$D\in\{0,1\}$.
\item If $W$ is uniform and independent of Alice's initial view,
$D$ is uniform and independent of her final view, even if she deviates.
\item If Alice is honest, every Bob strategy and every fixed
$\mathcal B\subseteq\{0,1\}^h$ satisfy
\begin{equation}
\Prb[W_0,W_1\in\mathcal B]
\leq16\frac{|\mathcal B|}{2^h}.
\label{eq:mal_ih}
\end{equation}
\end{enumerate}

\subsubsection{Clifford unitaries}

Let $\mathcal P_s=\{I,X,Y,Z\}^{\otimes s}$. A unitary $V$ is Clifford if, for every $P\in\mathcal P_s$,
$VPV^\dagger=\pm Q$ for some $Q\in\mathcal P_s$.

\begin{lemma}
\label{lem:mal_clifford}
Fix $0\leq u\leq s$ and choose $V$   
uniformly from the $s$-qubit
Clifford group modulo global phase. For $1\leq m\leq s$, define the channel
$\mathcal N_V(\rho)
=V\left(\rho\otimes\frac{I_{2^{s-m}}}{2^{s-m}}\right)V^\dagger$. Let $Y_Q=(Y_i)_{i\in Q}$ denote the output qubits of
$\mathcal N_V$, where $|Q|=s$.
Fix $\mathcal U\subseteq Q$ with $|\mathcal U|=u$. 
Except with probability at most $4^{m-u}$,
\begin{equation}
\operatorname{Tr}_{Y_{\mathcal U}}
\left[
(\mathcal N_V\otimes\operatorname{id}_R)(\rho^{XR})
\right]
=
\frac{I_{Y_{Q\setminus\mathcal U}}}{2^{s-u}}
\otimes\rho^R
\label{eq:mal_depolarizing}
\end{equation}
holds for every state $\rho^{XR}$, where $X$ consists of
$m$ qubits.
\end{lemma}

\begin{proof}
We have
\begin{align}
\Prb_V\!\left[
\exists P\in\mathcal P_m\setminus\{I\}:
\operatorname{Tr}_{Y_{\mathcal U}}
[V(P\otimes I)V^\dagger]\neq0
\right]\notag
&\leq
\sum_{P\in\mathcal P_m\setminus\{I\}}
\Prb_V\!\left[
\operatorname{Tr}_{Y_{\mathcal U}}
[V(P\otimes I)V^\dagger]\neq0
\right]\notag\\
&=
(4^m-1)\frac{4^{s-u}-1}{4^s-1} \notag\\
&
\leq 4^{m-u},
\label{eq:mal_pauli_failure}
\end{align}
where the equality holds because $V(P\otimes I)V^\dagger$ is uniform
over the $4^s-1$ nonidentity $s$-qubit Paulis modulo phase
\cite[Lem.~2--3]{webb2015clifford}, and exactly
$4^{s-u}-1$ have nonzero partial trace over $Y_{\mathcal U}$,
namely those with identity factors on every qubit in
$\mathcal U$.

Outside the event in \eqref{eq:mal_pauli_failure}, i.e., suppose  $\forall P\in\mathcal P_m\setminus\{I\}:
\operatorname{Tr}_{Y_{\mathcal U}}
[V(P\otimes I)V^\dagger]=0$,  for any state $\rho^{XR}$, we have 
\begin{align*}
\operatorname{Tr}_{Y_{\mathcal U}}
\left[
(\mathcal N_V\otimes\operatorname{id}_R)(\rho^{XR})
\right] & = \operatorname{Tr}_{Y_{\mathcal U}}
\left[
(\mathcal N_V\otimes\operatorname{id}_R)
\left(
\frac{1}{2^m}\sum_{P\in\mathcal P_m}
P\otimes
\operatorname{Tr}_X\!\left[(P\otimes I_R)\rho^{XR}\right]
\right)
\right]\\
&=
\frac1{2^s}\sum_{P\in\mathcal P_m}
\operatorname{Tr}_{Y_{\mathcal U}}
[V(P\otimes I)V^\dagger]
\otimes
\operatorname{Tr}_X[(P\otimes I_R)\rho^{XR}]\\
&=
\frac1{2^s}\operatorname{Tr}_{Y_{\mathcal U}}
[I_{Y_Q}]\otimes\rho^R\\
&=
\frac{I_{Y_{Q\setminus\mathcal U}}}{2^{s-u}}
\otimes\rho^R.
\end{align*}
Since the event in \eqref{eq:mal_pauli_failure} depends only
on $V$ and the fixed set $\mathcal U$, this proves
\eqref{eq:mal_depolarizing}  for every
$\rho^{XR}$.
\end{proof}

\subsection{Coding Scheme}
\label{subsec:malicious_scheme}
\label{subsec:malicious_finite_proof}

Fix $0<\epsilon<1$, $n\in\mathbb N$, and $0<\delta,\nu<1$.
Define $r\triangleq\left\lceil
\sqrt{n\ln\!\left(256/\min\{\delta,\nu\}\right)}
\right\rceil$ and choose
\begin{equation}
m\triangleq \max\left\{0,\left\lfloor
n\min\{\epsilon,1-\epsilon\}
-16\sqrt{n\ln\!\left(256/\min\{\delta,\nu\}\right)}
\right\rfloor\right\}.
\label{eq:mal_parameter_choice}
\end{equation}
For each integer $k\in\{0,\ldots,n\}$ satisfying
$|k-n(1-\epsilon)|\leq r$, define
$t=\lceil(k+m+4r)/2\rceil$ and fix a public injection
$\mathsf{Enc}_k:\{0,1\}^{h_k}\to\binom{[t]}{k-t}$, where
$h_k=\lfloor\log_2\binom{t}{k-t}\rfloor$.
For an ordered list $L$ and $\mathcal T\subseteq[t]$,
$L^{\mathcal T}$ denotes the entries in the indicated positions.
The protocol is given in Algorithm~\ref{alg:malicious_qot}.

\begin{algorithm}[!htbp]
\caption{Achievability scheme for quantum oblivious transfer over
$\mathcal E_\epsilon$ (malicious setting)}
\label{alg:malicious_qot}
\begingroup
\setstretch{1}

\begin{algorithmic}[1]
\Require Parameters chosen in \eqref{eq:mal_parameter_choice}, with $m\geq1$,
an input state $\Psi^{A_1A_2R}$, and Bob's choice
$C\in\{1,2\}$.

\Statex A failed check, missing or invalid message, or
interrupted hashing computation causes an abort with $F=1$.

\State \textit{Entanglement generation.}
\begin{algsubstates}

    \State For each $i\in[n]$, Alice prepares
    $$\ket{\Phi^+}^{U_iV_i}
    =\frac{1}{\sqrt{2}}(\ket{00}+\ket{11}).$$

    \State Alice keeps $U_i$ and sends  $V_i$ through
    $\mathcal E_\epsilon$.

\end{algsubstates}

\State \textit{Abort test.}
\begin{algsubstates}

    \State Bob defines
    $$\mathcal G \triangleq \{i\in[n]:\text{the $i$-th output is not erased}\}$$
    and sends $k\triangleq|\mathcal G|$ to Alice.

    \State Alice aborts unless Bob's announced
$k\in\{0,\ldots,n\}$ satisfies
\begin{equation}
    n(1-\epsilon)-r\leq k\leq n(1-\epsilon)+r. \label{eq:testk}
\end{equation}
    Define $t\triangleq\lceil(k+m+4r)/2\rceil$ and $s\triangleq2t-k$.

\end{algsubstates}

\State \textit{List selection.}
\begin{algsubstates}

    \State Bob samples $W$ uniformly from $\{0,1\}^{h_k}$
    and sets $\mathcal T\triangleq\mathsf{Enc}_k(W)$.

    \State Bob chooses $L_C$ uniformly among ordered lists
    of $t$ distinct elements of $\mathcal G$.

    \State Bob places the $k-t$ unused unerased indices
$\mathcal G\setminus L_C$ in the positions of $L_{\bar C}$
indexed by $\mathcal T$, in a uniformly random order.

\State Bob uniformly selects $s$ distinct erased indices from
$[n]\setminus\mathcal G$ and places them in the remaining
positions of $L_{\bar C}$, in a uniformly random order.

    \State Bob sends $(L_1,L_2)$ to Alice, who aborts unless
each list contains $t$ distinct indices from $[n]$ and
$L_1\cap L_2=\varnothing$.

\end{algsubstates}

\State \textit{Interactive hashing.}
\begin{algsubstates}

    \State The parties run interactive hashing with Bob's
    input $W$, obtaining $W_0,W_1$.

    \State Bob determines $D\in\{0,1\}$ such that $W_D=W$
    and sends  to Alice $$ a\triangleq D\oplus(C-1).$$ 

    \State The parties set
    $\mathcal T_i\triangleq\mathsf{Enc}_k(W_i)$ for $i\in\{0,1\}$,
    $J_1\triangleq L_1^{\mathcal T_{1-a}}$, and
    $J_2\triangleq L_2^{\mathcal T_a}$.

    \State For $i\in\{1,2\}$, the parties form $Q_i$ by removing
the indices in $J_i$ from $L_i$ and keeping the remaining
indices in their original order.

\end{algsubstates}

\State \textit{Testing.}
\begin{algsubstates}

   \State Bob measures the received halves $V_i$ of Alice's
Bell pairs, for $i\in J_1\cup J_2$, in the computational
basis and sends all outcomes to Alice.

   \State Alice measures each $U_i$, $i\in J_1\cup J_2$, in the
computational basis and aborts if any outcome differs from
the one Bob reported for the same index.

\end{algsubstates}

\State \textit{Teleportation.}
\begin{algsubstates}

    \State Alice samples independent uniform $s$-qubit Clifford
    unitaries $V^{(1)},V^{(2)}$ and sends their descriptions to Bob.

    \State For each $i\in\{1,2\}$, Alice applies $(V^{(i)})^T$
to $U_{Q_i}$, whose qubits are ordered according to $Q_i$, and  then performs a Bell measurement on the $j$-th qubit
of $A_i$ and the qubit in position $j$ of this transformed
register, for $j\in [m]$, as in
Section~\ref{sec:teleportation}.
Let $T_i\in\{0,1\}^{2m}$ denote the outcome and
$P_{T_i}$ its associated Pauli operator.

    \State Alice sends to Bob $(T_1,T_2)$.

\end{algsubstates}

\State \textit{Decoding.}
\begin{algsubstates}

   \State Bob applies $(V^{(C)})^\dagger$ to the received
Bell-pair halves $V_i$, $i\in Q_C$, ordered according to
$Q_C$. He then applies $P_{T_C}^\dagger$ to the qubits
in positions $1,\ldots,m$ of the resulting register.
    \State Bob outputs these $m$ qubits as $B_C$
    and sends $F=0$ to Alice.

\end{algsubstates}

\end{algorithmic}
\endgroup
\end{algorithm}

\subsection{Coding Scheme Analysis}
\label{subsec:malicious_analysis}

We verify Conditions~1--4 and construct the two simulators
required by Condition~5 of Definition~\ref{def:qot_delta_achievabilityb}.

\subsubsection{Abort probability and correctness}

In Line 3 of Algorithm \ref{alg:malicious_qot}, the construction of $L_c$ requires $t\leq k$ unerased
positions, while completing $L_{\bar c}$ requires
$s\leq n-k$ erased positions. We verify these two inequalities
for every accepted honest count $k=|\mathcal G|$. \eqref{eq:mal_parameter_choice} and \eqref{eq:testk} give
$\min\{k,n-k\}-m
\geq n\min\{\epsilon,1-\epsilon\}-r-m
\geq16(r-1)-r
\geq12r$, where the last inequality holds because for $m\geq 1$,  
$16\sqrt{n\ln(256/\min\{\delta,\nu\})}<n/2$,
which, together with $\ln(256/\min\{\delta,\nu\})>1$,
gives $n>1024$ and $r>32$. Consequently, $k-m\geq12r$ ensures
$t=\lceil(k+m+4r)/2\rceil\leq k-4r<k$,
so Bob can choose $L_c$.
Similarly, $n-k-m\geq12r$ gives
$s=2t-k\leq  n-k$,
so he can complete $L_{\bar c}$.
 
Fix $c\in\{1,2\}$ and suppose an honest 
$k=|\mathcal G|$ is accepted in Line~2 of
Algorithm~\ref{alg:malicious_qot}.
We verify that     $J_1\cup J_2\subseteq\mathcal G$ so that  every Bell pair tested in Line~5 is shared by Alice
and Bob, and their computational-basis outcomes agree
and the test passes.   Line~3 gives
$L_c\subseteq\mathcal G$, so
$J_c\subseteq L_c\subseteq\mathcal G$. Then, Line~3 also places all indices in
$\mathcal G\setminus L_c$ at positions $\mathcal T$
of $L_{\bar c}$, so that by Line~4,  
we obtain
$J_{\bar c}=L_{\bar c}^{\mathcal T_D}=L_{\bar c}^{\mathcal T}
=\mathcal G\setminus L_c\subseteq\mathcal G$.

  Hence, the protocol aborts
exactly when the  test in Line~2 fails:
\begin{align}
p_1^{(c)}(\Psi)
&=1-\Prb[n(1-\epsilon)-r\leq|\mathcal G|
                         \leq n(1-\epsilon)+r]
\nonumber\\
&\leq2e^{-2r^2/n}\nonumber\\
&
\leq2\left(\frac{\min\{\delta,\nu\}}{256}\right)^2\nonumber\\
&
\leq\delta,
\nonumber
\end{align}
where the first inequality follows from Hoeffding's
inequality and
$|\mathcal G|\sim\operatorname{Bin}(n,1-\epsilon)$,
and the second inequality follows from the definition of $r$.

It remains to verify recovery of $A_c$.
Line~3 of Algorithm~\ref{alg:malicious_qot} gives
$L_c\subseteq\mathcal G$, and Line~4 sets
$Q_c=L_c\setminus J_c$ so that  $Q_c\subseteq\mathcal G$, and Alice and Bob share
all Bell pairs indexed by $Q_c$, with $|Q_c|=t-(k-t)=s$ by Line~4.
The choice $t=\lceil(k+m+4r)/2\rceil$ in Line~2
then gives $s=2t-k\geq m+4r\geq m$, providing enough
pairs for the $m$ Bell measurements in Line~6. Then, by
\cite[Ex.~3.7.12]{wilde2013quantum},
Alice's application of $(V^{(c)})^T$ in Line~6 is
equivalent to applying $V^{(c)}$ to the receiving halves.
In Line~7, Bob cancels this unitary and applies the
teleportation Pauli correction and 
Lemma~\ref{lem:teleportation_properties_ach} therefore gives
$\omega^{(c)}_{B_cR\mid F=0}(\Psi)
=(\operatorname{id}_{A_c\to B_c}\otimes\operatorname{id}_R)
(\Psi^{A_cR})$.

\subsubsection{Non-chosen state privacy}

We show that replacing the non-chosen input preserves
both Bob's conditional output states and the
probabilities with which these states are averaged.

\begin{lemma}
\label{lem:mal_input_independent_transcript}
In an honest execution of
Algorithm~\ref{alg:malicious_qot}, for each fixed
choice $c$, the joint distribution of
$(\mathcal G,M_0,F)$ is independent of the input
state $\Psi^{A_1A_2R}$.
\end{lemma}

\begin{proof}
Let $M_{\mathrm{pre}}$ consist of $k$ from Line~2,
$(L_1,L_2)$ from Line~3, the interactive-hashing
messages and bit $a$ from Line~4, the announced test
outcomes from Line~5, and the Clifford descriptions
from Line~6, so that  $M_0=(M_{\mathrm{pre}},T_1,T_2)$
on success.
With $\Pi_{t_i}$   the Bell projector
corresponding to outcome $t_i$ in Line~6, and  $m_0=(\mu,t_1,t_2)$, we have
\begin{align*}
&\Prb[
\mathcal G=\mathcal T,M_0=m_0,F=0\mid C=c]
\\
&=
\Prb[
\mathcal G=\mathcal T,M_{\mathrm{pre}}=\mu,F=0
\mid C=c]
\times
\Prb[
T_1=t_1,T_2=t_2
\mid \mathcal G=\mathcal T,M_{\mathrm{pre}}=\mu,F=0,C=c]
\\
&\overset{(a)}{=}
\Prb[
\mathcal G=\mathcal T,M_{\mathrm{pre}}=\mu,F=0
\mid C=c]
\times
\operatorname{Tr}\!\left[
(\Pi_{t_1}\otimes\Pi_{t_2})
\left(
\Psi^{A_1A_2}\otimes
\frac{I_{2^{2m}}}{2^{2m}}
\right)
\right]
\\
&\overset{(b)}{=}
4^{-2m}
\Prb[
\mathcal G=\mathcal T,M_{\mathrm{pre}}=\mu,F=0
\mid C=c], \numberthis \label{eq:sym2}
\end{align*}
where
\begin{enumerate}[(a)]
\item holds by identifying the state entering the Bell
measurements in Line~6. Let $H_i$ denote the $m$ qubits of
$U_{Q_i}$ measured with $A_i$ in Line~6.
Conditioned on $\mathcal G=\mathcal T$,
$M_{\mathrm{pre}}=\mu$, $F=0$, and $C=c$, this state satisfies
\begin{align*}
\rho^{A_1A_2H_1H_2}
&=
\Psi^{A_1A_2}\otimes
\bigotimes_{i=1}^{2}
\operatorname{Tr}_{U_{Q_i}\setminus H_i}
\left[
(V^{(i)})^T
\frac{I_{U_{Q_i}}}{2^s}
(V^{(i)})^*
\right]
\\
&=
\Psi^{A_1A_2}\otimes
\bigotimes_{i=1}^{2}
\operatorname{Tr}_{U_{Q_i}\setminus H_i}
\left[\frac{I_{U_{Q_i}}}{2^s}\right]
\\
&=
\Psi^{A_1A_2}\otimes
\frac{I_{H_1H_2}}{2^{2m}},
\end{align*}
where the first equality holds because Line~1 prepares independent
Bell pairs, independently of $\Psi$, and the tests
in Line~5 act on   $J_1\cup J_2$ with $(J_1\cup J_2)\cap(Q_1\cup Q_2)
=\varnothing$, so that
 even after conditioning on the test outcomes,
Alice's registers $U_{Q_1}U_{Q_2}$  have state
$I_{U_{Q_1}}/2^s\otimes I_{U_{Q_2}}/2^s$ before the Clifford operation;

\item holds because  $\operatorname{Tr}_{H_i}[\Pi_{t_i}]=\frac{I_{A_i}}{2^m}$ and
$\operatorname{Tr}[\Psi^{A_1A_2}]=1$.
\end{enumerate}
Lines~1--5 and the Clifford sampling in Line~6 use
neither input, and, by correctness, honest success is
determined in Line~2,
thus, \eqref{eq:sym2} gives joint probabilities
independent of $\Psi$ on success.
Since an honest abort also occurs in Line~2 before
either input is used, the full distribution of
$(\mathcal G,M_0,F)$ is independent of $\Psi$.
\end{proof}

\begin{lemma}
\label{lem:mal_conditional_nonchosen_privacy}
Fix $c\in\{1,2\}$.
In an honest execution of
Algorithm~\ref{alg:malicious_qot}, conditioned on
$M_0=m_0,F=0,\mathsf S=[n]\setminus\mathcal G$, the state on $B_cE_cR$
has the form
$\Psi^{B_cR}\otimes\xi_{\mathcal G,m_0}^{E_c}$,
where $\xi_{\mathcal G,m_0}^{E_c}$ is normalized
and independent of $\Psi$.
\end{lemma}

\begin{proof}
With the lists $Q_1,Q_2$, Clifford unitaries
$V^{(1)},V^{(2)}$, and Bell outcomes $T_1,T_2$
determined by $m_0$, we have
\begin{align}
\omega^{(c)}_{B_cE_cR
\mid M_0=m_0,\,F=0,\,\mathsf S=[n]\setminus\mathcal G}(\Psi)
\nonumber
&\overset{(a)}{=}
\operatorname{Tr}_{W_{Q_{\bar c}}}
\left[
\left(
\operatorname{id}_{A_c\to B_c}
\otimes\mathcal N_{V^{(\bar c)}}
\otimes\operatorname{id}_R
\right)
\left(
P_{T_{\bar c}}\Psi^{A_cA_{\bar c}R}
P_{T_{\bar c}}^\dagger
\right)
\right]
\otimes\xi_{\mathcal G,m_0}^{E_c}
\nonumber\\
&\overset{(b)}{=}
\left(
\operatorname{id}_{A_c\to B_c}
\otimes\operatorname{id}_R
\right)(\Psi^{A_cR})
\otimes\xi_{\mathcal G,m_0}^{E_c}
\nonumber\\
&=
\Psi^{B_cR}\otimes\xi_{\mathcal G,m_0}^{E_c},\nonumber
\end{align}
where
\begin{enumerate}[(a)]
\item holds because Lines~3--4 place the receiving
halves at $Q_c$ with Bob and those at $Q_{\bar c}$
in $W_{Q_{\bar c}}$, while Line~5 tests only the
removed positions $J_1\cup J_2$, leaving the pairs
at $Q_1\cup Q_2$ unchanged.
Thus, for each input $A_i$, with $i=c$ for the
chosen input and $i=\bar c$ for the non-chosen input,
Line~6 induces the following channel to the
receiving qubits at $Q_i$:
\begin{align*}
\rho^{A_i}
\longmapsto
V^{(i)}
\left(
P_{T_i}\rho^{A_i}P_{T_i}^{\dagger}
\otimes\frac{I_{2^{s-m}}}{2^{s-m}}
\right)
(V^{(i)})^{\dagger}
=
\mathcal N_{V^{(i)}}
\left(P_{T_i}\rho^{A_i}P_{T_i}^{\dagger}\right),
\end{align*}
where \eqref{eq:tele} transfers the input, up to
$P_{T_i}$, to the first $m$ receiving qubits,
tracing out Alice's $s-m$ unmeasured halves gives
the maximally mixed factor, and the transpose identity
\cite[Ex.~3.7.12]{wilde2013quantum}
shows that applying $(V^{(i)})^T$ to Alice's halves
of the $s$ Bell pairs produces the same joint state
as applying $V^{(i)}$ to their receiving halves.

Then, for the chosen block, Line~7 removes the Clifford
and Pauli, leaving $A_c$ in $B_c$ and $s-m$
independent maximally mixed qubits in $E_c$.
For the non-chosen block, the channel output remains
in $W_{Q_{\bar c}}$ and is traced out.
Consequently, the remaining $s-m$ receiving qubits at $Q_c$,
which Bob retains in $E_c$ after Line~7, his other
systems from Lines~1--5, and his retained records
fixed by $m_0$ form
$\xi_{\mathcal G,m_0}^{E_c}$, independently of $\Psi$;

\item holds because trace preservation of
$\mathcal N_{V^{(\bar c)}}$ reduces the trace over
$W_{Q_{\bar c}}$ to
$\operatorname{Tr}_{A_{\bar c}}
[P_{T_{\bar c}}\Psi^{A_cA_{\bar c}R}
P_{T_{\bar c}}^\dagger]=\Psi^{A_cR}$.
\end{enumerate}
\end{proof}

We now combine the two lemmas to prove
\eqref{eq:qot_global_destructive_privacy_general}.
For every state $\rho^{A_c}$, define
$\mathsf{Sim}_{c,n}(\rho^{A_c})
\triangleq
\omega^{(c)}_{B_cE_cM_0F}
\left(
\rho^{A_c}\otimes
\ket{0^m}\!\bra{0^m}^{A_{\bar c}}
\right)$,
which appends the fixed non-chosen input,
executes both honest parties and the erasure channel
in Algorithm~\ref{alg:malicious_qot} with choice $c$,
and returns $B_cE_cM_0F$.
$\mathsf{Sim}_{c,n}$ is CPTP and independent
of $\Psi$ because appending a fixed state, executing the protocol, and tracing out the other registers
are CPTP operations  independent of $\Psi$. Then, we have
\begin{align}
\omega^{(c)}_{B_cE_cM_0FR}
\left(\Psi^{A_1A_2R}\right)
&\overset{(a)}{=}
\omega^{(c)}_{B_cE_cM_0FR}
\left(
\Psi^{A_cR}\otimes
\ket{0^m}\!\bra{0^m}^{A_{\bar c}}
\right)
\nonumber\\
&\overset{(b)}{=}
(\mathsf{Sim}_{c,n}\otimes\operatorname{id}_R)
\left(\Psi^{A_cR}\right),\nonumber
\end{align}
where
\begin{enumerate}[(a)]
\item holds by comparing two honest executions of
Algorithm~\ref{alg:malicious_qot}, both with choice $c$:
one with input $\Psi^{A_1A_2R}$ and the other with input
$\Psi^{A_cR}\otimes
\ket{0^m}\!\bra{0^m}^{A_{\bar c}}$. By Lemma~\ref{lem:mal_input_independent_transcript},
these inputs give the same distribution of
$(\mathcal G,M_0,F)$.
It therefore suffices to show that, conditioned on
each value of these registers of positive probability,
the state on $B_cE_cR$ is identical for both inputs. For $F=0$, the conditional states on $B_cE_cR$ are
identical by Lemma~\ref{lem:mal_conditional_nonchosen_privacy},
since both inputs have the same marginal $\Psi^{A_cR}$.

For $F=1$, the conditional states on $B_cE_cR$ are
identical because the protocol terminates in Line~2
before using either input, leaving the same state
on $B_cE_c$ in tensor product with  $\Psi^R$;
\item holds by the definition of $\mathsf{Sim}_{c,n}$.
\end{enumerate}
\subsubsection{Receiver's choice privacy}

We first establish the symmetry of the list announcement in Line 3 and
interactive hashing. 

\begin{lemma}
\label{lem:mal_list_hash_symmetry}
Suppose that Bob is honest, while Alice may deviate.
For every accepted $k$, $w\in\{0,1\}^{h_k}$, and
$c\in\{1,2\}$, conditioned on
$|\mathcal G|=k,W=w,C=c$, the announcement
$(L_1,L_2)$ is uniform over pairs of disjoint ordered
$t$-lists, each with probability $1/(n)_{2t}$,
where $(n)_j=n!/(n-j)!$.

For an interactive-hashing transcript $\gamma$
with outputs $w_0,w_1$, let $q_\gamma(w)$ be the probability of obtaining
$\gamma$ when Bob's input is $w$ and Alice sends
the messages prescribed by $\gamma$.
Then, $q_\gamma(w_0)=q_\gamma(w_1)$. Moreover, for any fixed quantum strategy of Alice,
the subnormalized map associated with $\gamma$
during interactive hashing in Line~4 of
Algorithm~\ref{alg:malicious_qot}, from her registers
before hashing to all registers she retains afterward,
is identical for Bob's inputs $w_0$ and $w_1$.
\end{lemma}

\begin{proof}
We first consider the list announcement in Line~3:
\begin{align}
\Prb[(L_1,L_2)=(l_1,l_2)
\mid |\mathcal G|=k,W=w,C=c]
&=
\frac{1}{\binom nk}
\frac{1}{(k)_t}
\frac{1}{(k-t)!}
\frac{1}{(n-k)_{2t-k}}
\nonumber\\
&=\frac{1}{(n)_{2t}},
\label{eq:mal_list_uniform}
\end{align}
where the first equality holds because announcing $(l_1,l_2)$
requires 
$\mathcal G=l_c\cup l_{\bar c}^{\mathsf{Enc}_k(w)}$, which has probability $1/\binom nk$, since
$\mathcal G$ is uniform given $k$, even for an
arbitrary state transmitted by Alice in Line~1.
Conditioned on this set, the remaining three factors
are the probabilities of Line~3's uniform choices:
the ordered list $L_c$, the ordering of the remaining
$k-t$ unerased positions, and the ordered selection
of $2t-k$ erased positions, respectively.

Let Bob sample $W$ uniformly, while Alice sends the messages specified
by $\gamma$.
Let $\Gamma$ denote the full sequence of messages
 exchanged during interactive hashing. Then, we have
\begin{align}
2^{-h_k}q_\gamma(w_0)
&=
\Prb[W=w_0]\Prb[\Gamma=\gamma\mid W=w_0]
\nonumber\\
&=\Prb[\Gamma=\gamma,W=w_0]
\nonumber\\
&=\Prb[\Gamma=\gamma,D=0]
\nonumber\\
&=\Prb[\Gamma=\gamma,D=1]
\nonumber\\
&=2^{-h_k}q_\gamma(w_1),
\label{eq:mal_hash_equal_weights}
\end{align}
where the first equality holds by uniformity of $W$ and definition of $q_\gamma(w_0)$, and the penultimate  equality holds because, by Property~2 of interactive hashing, $D$ is uniform and independent of Alice's view,
which includes $\Gamma$.

For a fixed quantum strategy of Alice, let
$\mathcal A_\gamma$ denote the completely positive,
trace-nonincreasing map describing her operations
in Line~4 when she receives and sends the messages
in $\gamma$, with all her retained registers included
in the output.
This map is independent of Bob's input $w$.
Bob's replies contribute the probability $q_\gamma(w)$,
so the subnormalized map for transcript $\gamma$ is
$q_\gamma(w)\mathcal A_\gamma$, and  \eqref{eq:mal_hash_equal_weights} thus gives
$q_\gamma(w_0)\mathcal A_\gamma
=q_\gamma(w_1)\mathcal A_\gamma$.
\end{proof}

We now show \eqref{eq:qot_receiver_privacy2}. First consider $F=0$.
Let $H_i$ denote the $m$ qubits of $U_{Q_i}$
measured with $A_i$ in Line~6, and set
$\Pi=\Pi_{t_1}\otimes\Pi_{t_2}\otimes I_R$.
Let $\rho_{\mathrm{pre}}$ denote the state on
$A_1A_2RH_1H_2$ immediately before these measurements,
conditioned on $C=c$, $\mathsf S=\mathcal S$, and
the messages in $m_0$ announced before the measurements.
Write
$\tau_{Q_A\setminus A_1A_2H_1H_2,m_0}$
for the normalized conditional state of Alice's
other retained registers, including the fixed
Bell-outcome records. Whenever the conditioning event has positive
probability,
\begin{align}
\omega^{(c)}_{Q_AR
\mid M_0=m_0,F=0,\mathsf S=\mathcal S}(\Psi)
\nonumber
&\overset{(a)}{=}
\frac{\Pi\rho_{\mathrm{pre}}\Pi}
{\operatorname{Tr}[\Pi\rho_{\mathrm{pre}}]}
\otimes
\tau_{Q_A\setminus A_1A_2H_1H_2,m_0}
\nonumber\\
&\overset{(b)}{=}
\frac{
\Pi\left[
\Psi^{A_1A_2R}\otimes
\displaystyle\bigotimes_{i=1}^{2}
\operatorname{Tr}_{U_{Q_i}\setminus H_i}
\left[
(V^{(i)})^T
\frac{I_{U_{Q_i}}}{2^s}
(V^{(i)})^*
\right]
\right]\Pi
}{
\operatorname{Tr}[\Pi\rho_{\mathrm{pre}}]
}
\otimes
\tau_{Q_A\setminus A_1A_2H_1H_2,m_0}
\nonumber\\
&=
\frac{
\Pi\left[
\Psi^{A_1A_2R}\otimes
\frac{I_{H_1H_2}}{2^{2m}}
\right]\Pi
}{
\operatorname{Tr}[\Pi\rho_{\mathrm{pre}}]
}
\otimes
\tau_{Q_A\setminus A_1A_2H_1H_2,m_0}
\nonumber\\
&\overset{(c)}{=}
\frac{
4^{-2m}
(\Pi_{t_1}\otimes\Pi_{t_2})\otimes\Psi^R
}{
4^{-2m}
}
\otimes
\tau_{Q_A\setminus A_1A_2H_1H_2,m_0}
\nonumber\\
&\overset{(d)}{=}
\tau_{Q_A,m_0,0}\otimes\Psi^R
\nonumber\\
&\triangleq\tau_{Q_ARm_0\,0}(\Psi),
\nonumber
\end{align}
where
\begin{enumerate}[(a)]
\item holds because, with $\rho$ the state immediately
after Line~1, we have
\begin{equation}
\rho^{A_1A_2RU^n\mid C=c,\mathsf S=\mathcal S}
=
\Psi^{A_1A_2R}\otimes\frac{I_{U^n}}{2^n},
\label{eq:mal_alice_source}
\end{equation}
since the Bell pairs are prepared independently
of $\Psi$ and each of Alice's halves is maximally mixed.
Lines~2--4 leave these quantum registers untouched,
the tests in Line~5 act outside $Q_1\cup Q_2$,
and the Clifford operations in Line~6 preserve
the maximally mixed state on $U_{Q_1}U_{Q_2}$.
Hence, immediately before the Bell measurements,
$A_1A_2RH_1H_2$ is independent of Alice's other registers. The Bell measurements act only on $A_1H_1$ and
$A_2H_2$, so conditioning on their outcomes replaces
$\rho_{\mathrm{pre}}$ by
$\Pi\rho_{\mathrm{pre}}\Pi/
\operatorname{Tr}[\Pi\rho_{\mathrm{pre}}]$
and leaves the other registers unchanged. Bob's decoding in Line~7 leaves this state on
$Q_AR$ unchanged;

\item holds because Lines~2--4 leave
$A_1A_2U^n$ untouched, and Line~5 tests only
$J_1\cup J_2$, which is disjoint from $Q_1\cup Q_2$.
By \eqref{eq:mal_alice_source} and independence
of the Bell pairs, conditioning on these tests
therefore leaves the state on
$A_1A_2RU_{Q_1}U_{Q_2}$ equal to
$\Psi^{A_1A_2R}\otimes
I_{U_{Q_1}}/2^s\otimes I_{U_{Q_2}}/2^s$.
Line~6 applies $(V^{(i)})^T$ to $U_{Q_i}$;
\item holds because 
$\operatorname{Tr}_{H_i}[\Pi_{t_i}]=I_{A_i}/2^m$, and the trace of the numerator is $4^{-2m}$;

\item holds by defining
$\tau_{Q_A,m_0,0}\triangleq
(\Pi_{t_1}\otimes\Pi_{t_2})
\otimes\tau_{Q_A\setminus A_1A_2H_1H_2,m_0}$, which is independent of $c$, $\mathcal S$, and
$\Psi$. Indeed, Steps~$(a)$--$(c)$ show that the state on
$A_1A_2U^n$ and Alice's copies of $(M_0,F)$
depends only on $m_0$.
For her private hashing registers,
Lemma~\ref{lem:mal_list_hash_symmetry} and
\eqref{eq:mal_alice_source} give the same
conditional state for all compatible $c,\mathcal S$.
This private state is also independent of $\Psi$,
since honest hashing in Line~4 acts only on
classical registers initialized independently
of the inputs.
Thus, $\tau_{Q_A,m_0,0}$ depends only on $m_0$.
\end{enumerate}

When $F=1$, the honest execution terminates in
Line~2 before either input is used.~\eqref{eq:mal_alice_source} therefore gives
$\omega^{(c)}_{Q_AR\mid M_0=m_0,F=1,\mathsf S=\mathcal S}(\Psi)
=\Psi^{A_1A_2R}\otimes
\tau_{Q_A\setminus A_1A_2,m_0,1}
\triangleq\tau_{Q_ARm_0\,1}(\Psi)$,
where the remaining-register state depends only
on $m_0$.

We now show  \eqref{eq:qot_receiver_privacy1}.
Let $\Gamma$ denote the interactive-hashing transcript
in Line~4.
For a successful public transcript $m_0$, denote its
announced count, lists, hashing transcript, bit $a$,
and Clifford unitaries by
$k,(l_1,l_2),\gamma,\alpha,(v_1,v_2)$, respectively.
Then, we have
\begin{align*}
\Prb[M_0=m_0,F=0\mid C=c]
&=
\Prb[|\mathcal G|=k\mid C=c]\,
\Prb[L_1=l_1,L_2=l_2\mid|\mathcal G|=k,C=c]
\\
&\qquad{}\times
\Prb[\Gamma=\gamma,a=\alpha
\mid|\mathcal G|=k,L_1=l_1,L_2=l_2,C=c]
\\
&\qquad{}\times
\Prb[M_0=m_0,F=0
\mid|\mathcal G|=k,L_1=l_1,L_2=l_2,
\Gamma=\gamma,a=\alpha,C=c]
\\
&\overset{(a)}{=}
\frac{\Prb[|\mathcal G|=k]}{(n)_{2t}}
\Prb[\Gamma=\gamma,D=\alpha\oplus(c-1)
\mid|\mathcal G|=k,L_1=l_1,L_2=l_2,C=c]
\\
&\qquad{}\times
\Prb[M_0=m_0,F=0
\mid|\mathcal G|=k,L_1=l_1,L_2=l_2,
\Gamma=\gamma,a=\alpha,C=c]
\\
&\overset{(b)}{=}
\frac{\Prb[|\mathcal G|=k]}{(n)_{2t}}
\Prb[\Gamma=\gamma,D=\alpha\oplus(c-1)
\mid|\mathcal G|=k,L_1=l_1,L_2=l_2,C=c]
\\
&\qquad{}\times
2^{-2(k-t)}
\Prb[V^{(1)}=v_1,V^{(2)}=v_2
\mid|\mathcal G|=k]\,
4^{-2m}
\\
&\overset{(c)}{=}
\frac{\Prb[|\mathcal G|=k]}{2(n)_{2t}}
\Prb[\Gamma=\gamma
\mid|\mathcal G|=k,L_1=l_1,L_2=l_2,C=c]
\\
&\qquad{}\times
2^{-2(k-t)}
\Prb[V^{(1)}=v_1,V^{(2)}=v_2
\mid|\mathcal G|=k]\,
4^{-2m}
\\
&\overset{(d)}{=}
\frac{\Prb[|\mathcal G|=k]}{2(n)_{2t}}
\Prb[\Gamma=\gamma
\mid|\mathcal G|=k,L_1=l_1,L_2=l_2,C=\bar c]
\\
&\qquad{}\times
2^{-2(k-t)}
\Prb[V^{(1)}=v_1,V^{(2)}=v_2
\mid|\mathcal G|=k]\,
4^{-2m}
\\
&\overset{(e)}{=}
\Prb[M_0=m_0,F=0\mid C=\bar c],
\end{align*}
where
\begin{enumerate}[(a)]
\item holds because $|\mathcal G|$ is independent
of $C$ in Line~1, averaging \eqref{eq:mal_list_uniform} over $W$ gives
$\Prb[L_1=l_1,L_2=l_2\mid|\mathcal G|=k,C=c]
=1/(n)_{2t}$, and Line~4 sets
$a=D\oplus(C-1)$;

\item holds because the factor $2^{-2(k-t)}$ is
the probability of the reported test string:
by correctness, $J_1\cup J_2\subseteq\mathcal G$,
so Line~5 measures $2(k-t)$ independent Bell pairs.
The Clifford probability is conditioned only on
$|\mathcal G|=k$ because Line~6 samples these
unitaries independently given $k$.
The factor $4^{-2m}$ is the Bell-outcome probability:
Line~5 leaves the pairs at $Q_1\cup Q_2$ unchanged,
and moving the Cliffords to their receiving halves
by the transpose identity
\cite[Ex.~3.7.12]{wilde2013quantum} allows
Lemma~\ref{lem:teleportation_properties_ach}
to be applied successively to the teleportations of $A_1$ and $A_2$
in Line~6, each contributing a factor $4^{-m}$.
The success-probability factor is $1$, since
correctness gives $F=0$ for accepted $k$;

\item holds because applying Property~2 to the hashing
execution, with $|\mathcal G|=k$, $L_1=l_1$,
$L_2=l_2$, and $C=c$ fixed, gives
$\Prb[D=\alpha\oplus(c-1)
\mid\Gamma=\gamma,|\mathcal G|=k,
L_1=l_1,L_2=l_2,C=c]=1/2$. Property~2 can indeed be applied because, before Line~4, Bayes' rule and
\eqref{eq:mal_list_uniform} give
$\Prb[W=w\mid|\mathcal G|=k,L_1=l_1,L_2=l_2,C=c]
=(2^{-h_k}/(n)_{2t})/(1/(n)_{2t})=2^{-h_k}$,
and $W$ is independent of Alice's quantum state:
\eqref{eq:mal_alice_source} gives the same state
for every erasure pattern, and Lines~2--3 use only
the erasures and independent randomness;

\item holds because Line~1 gives the same state on
$A_1A_2U^n$ for $C=1$ and $C=2$, and Lines~2--3
leave these registers unchanged, while $(c)$
establishes that $W$ is uniform and independent
of them under the stated conditioning.
The hashing messages $\Gamma$ in Line~4 therefore
have the same distribution for $C=1$ and $C=2$,
since they are generated from the same joint state
by a procedure that does not use $C$;

\item holds by $(a)$--$(d)$ with $c$ replaced
by $\bar c$.
\end{enumerate}
When $F=1$, Line~2 determines the public messages
and abort decision solely from $k$, whose distribution
is independent of $C$, hence
$\Prb[M_0=m_0,F=1\mid C=1]
=\Prb[M_0=m_0,F=1\mid C=2]$.

\subsubsection{Simulation against a dishonest Bob}

After Bob announces valid $k,L_1,L_2$ in Lines~2--3,
let $u_i\triangleq|L_i\setminus\mathcal G|$ count
the erased channel uses indexed by $L_i$.
Choose $j\in\operatorname*{arg\,max}_{i\in\{1,2\}}u_i$,
breaking ties by a fixed rule.

The ideal experiment allows Bob to receive at most
one of Alice's inputs. We therefore aim to reproduce
his final private register $V$ and the public
communication $M$ using his initial register $Z$
and only $A_{3-j}$.
To this end, we compare the real execution
with one that replaces $A_j$ by
$\ket{0^m}\!\bra{0^m}$, preserving the marginal on
$A_{3-j}ZR$, and bound the trace distance between
their final states on $VMR$.

The following lemma first bounds the probability
that Bob's reported outcomes match Alice's
measurement outcomes on $J_1\cup J_2$ in Line~5
and $|Q_j\setminus\mathcal G|<m+r^2/n$.

\begin{lemma}
\label{lem:mal_remaining_erasure_bound}
For every Bob strategy and every initial state,
\begin{align*}
\Prb[\text{Line~6 is reached},\
|Q_j\setminus\mathcal G|<m+r^2/n]
\leq
\Prb[|\mathcal G|>n(1-\epsilon)+r]
+64e^{-r^2/n}+2^{-2r^2/n}.
\end{align*}
\end{lemma}

\begin{proof}
Fix $\mathcal G$ and valid $k,L_1,L_2$ before
Line~4, and suppose that $u_1,u_2\geq r$.
Let $\mathcal T$ be uniform in $\binom{[t]}{k-t}$
and define
$\mathcal B\triangleq
\{w:\min_i|L_i^{\mathsf{Enc}_k(w)}
\setminus\mathcal G|<2r^2/n\}$.
Conditioned on these fixed values, we have
\begin{align}
\Prb[\text{hashing completes},\
W_0,W_1\in\mathcal B]
\nonumber
&\overset{(a)}{\leq}
16\frac{|\mathcal B|}{2^{h_k}}
\nonumber\\
&\overset{(b)}{=}
16\frac{
\left|\{\mathsf{Enc}_k(w):w\in\mathcal B\}\right|
}{2^{h_k}}
\nonumber\\
&\overset{(c)}{\leq}
16\frac{\binom{t}{k-t}}{2^{h_k}}
\sum_{i=1}^{2}
\Prb[|L_i^{\mathcal T}\setminus\mathcal G|<2r^2/n]
\nonumber\\
&\overset{(d)}{\leq}
32\sum_{i=1}^{2}
\Prb[|L_i^{\mathcal T}\setminus\mathcal G|<2r^2/n]
\nonumber\\
&\overset{(e)}{\leq}
32\sum_{i=1}^{2}
\Prb\left[
|L_i^{\mathcal T}\setminus\mathcal G|
<\frac{(k-t)u_i}{2t}
\right]
\nonumber\\
&\overset{(f)}{\leq}
32\sum_{i=1}^{2}
\exp\left(-\frac{(k-t)u_i}{8t}\right)
\nonumber\\
&\overset{(g)}{\leq}
64e^{-r^2/n},
\label{eq:mal_hash_failure}
\end{align}
where
\begin{enumerate}[(a)]
\item holds by \eqref{eq:mal_ih} because a classical
randomized Bob can reproduce a quantum Bob's
reply probabilities in Line~4, and extending
an interrupted hashing execution to completion
can only enlarge the event;

\item holds by injectivity of $\mathsf{Enc}_k$;

\item holds by the definition of $\mathcal B$
and the union bound;

\item holds because
$h_k=\lfloor\log_2\binom{t}{k-t}\rfloor$;

\item holds because
$(k-t)u_i/t\geq8r^2/n$:
the correctness analysis gives
$k-t=\lfloor(k-m-4r)/2\rfloor\geq4r$
for every accepted $k$, while $u_i\geq r$
and valid disjoint lists satisfy $2t\leq n$;

\item holds by the hypergeometric lower-tail bound
\cite[Cor.~4.3 and Sec.~5.2]{mulzer2018five},
since $|L_i^{\mathcal T}\setminus\mathcal G|$
has mean $(k-t)u_i/t$;

\item holds by the bound $(k-t)u_i/t\geq8r^2/n$
established in $(e)$.
\end{enumerate}

Suppose now that Bob's reported outcomes match Alice's
measurement outcomes on $J_1\cup J_2$ in Line~5 of
Algorithm~\ref{alg:malicious_qot}, and that
$\min\{u_1,u_2\}<r$,
$|(J_1\cup J_2)\setminus\mathcal G|<2r^2/n$,
and $|\mathcal G|\leq n(1-\epsilon)+r$.
Then, we have
\begin{align}
|Q_j\setminus\mathcal G|
&=u_j-|J_j\setminus\mathcal G|
\nonumber\\
&\overset{(a)}{\geq}
u_1+u_2-r-|(J_1\cup J_2)\setminus\mathcal G|
\nonumber\\
&\overset{(b)}{\geq}
2t-|\mathcal G|-r-\frac{2r^2}{n}
\nonumber\\
&\overset{(c)}{\geq}
m+3r+k-|\mathcal G|-\frac{2r^2}{n}
\nonumber\\
&\overset{(d)}{\geq}
m+r-\frac{2r^2}{n}
\nonumber\\
&\overset{(e)}{\geq}
m+\frac{r^2}{n},
\label{eq:mal_remaining_erasures}
\end{align}
where
\begin{enumerate}[(a)]
\item holds because
$u_j=u_1+u_2-\min\{u_1,u_2\}>u_1+u_2-r$
and $J_j\subseteq J_1\cup J_2$;
\item holds because the disjoint lists give
$u_1+u_2=2t-|(L_1\cup L_2)\cap\mathcal G|
\geq2t-|\mathcal G|$;
\item holds by
$t=\lceil(k+m+4r)/2\rceil$ in Line~2;
\item holds because Line~2 accepts only
$k\geq n(1-\epsilon)-r$, so
$k-|\mathcal G|\geq-2r$;
\item holds because $r\leq n/3$.
Indeed, \eqref{eq:mal_parameter_choice} and $m\geq1$
give $16\sqrt{n\ln(256/\min\{\delta,\nu\})}<n/2$
and $n>1024$, whence $r<n/32+1\leq n/3$.
\end{enumerate}

Removing the conditioning on $\mathcal G,k,L_1,L_2$,
we obtain
\begin{align*}
\Prb[\text{Line~6 is reached},\
|Q_j\setminus\mathcal G|<m+r^2/n]
&\overset{(a)}{\leq}
\Prb[|\mathcal G|>n(1-\epsilon)+r]
\\
&\qquad+
\Prb[\text{hashing completes},\
\min\{u_1,u_2\}\geq r,\ W_0,W_1\in\mathcal B]
\\
&\qquad+
\Prb[\text{tests pass},\
|(J_1\cup J_2)\setminus\mathcal G|\geq2r^2/n]
\\
&\overset{(b)}{\leq}
\Prb[|\mathcal G|>n(1-\epsilon)+r]
+64e^{-r^2/n}+2^{-2r^2/n},
\end{align*}
where
\begin{enumerate}[(a)]
\item holds by the union bound applied to the events
 $\mathcal E_1\triangleq
\{|\mathcal G|>n(1-\epsilon)+r\}$, $\mathcal E_2\triangleq
\{\text{hashing completes},\
\min\{u_1,u_2\}\geq r,\ W_0,W_1\in\mathcal B\}$, $\mathcal E_3\triangleq
\{\text{tests in Line~5 pass},\
|(J_1\cup J_2)\setminus\mathcal G|\geq2r^2/n\}$. To see that $\{\text{Line~6 is reached},\
|Q_j\setminus\mathcal G|<m+r^2/n\}
\subseteq\mathcal E_1\cup\mathcal E_2\cup\mathcal E_3$, suppose Line~6 is reached
but none of $\mathcal E_1$, $\mathcal E_2$, $\mathcal E_3$
occurs.
The tests in Line~5 have passed, so
$\mathcal E_3^c$ implies
$|(J_1\cup J_2)\setminus\mathcal G|<2r^2/n$.
Since Line~4 sets
$J_1=L_1^{\mathsf{Enc}_k(W_{1-a})}$ and
$J_2=L_2^{\mathsf{Enc}_k(W_a)}$,
the inequalities
$|J_i\setminus\mathcal G|<2r^2/n$, $i=1,2$,
give $W_{1-a},W_a\in\mathcal B$
by the definition of $\mathcal B$,
so $\mathcal E_2^c$ implies
$\min\{u_1,u_2\}<r$.
Together with
$|\mathcal G|\leq n(1-\epsilon)+r$
from $\mathcal E_1^c$,
\eqref{eq:mal_remaining_erasures} therefore gives
$|Q_j\setminus\mathcal G|\geq m+r^2/n$.
\item holds by \eqref{eq:mal_hash_failure}, averaged over the
erasures and valid announcements, and without assuming $u_1,u_2\geq r$, the tests in
Line~5 satisfy
$\Prb[\text{tests pass},\
|(J_1\cup J_2)\setminus\mathcal G|\geq2r^2/n]
\leq
2^{-2r^2/n}
\Prb[|(J_1\cup J_2)\setminus\mathcal G|\geq2r^2/n]
 \leq2^{-2r^2/n}$, 
where the first inequality holds because, conditioned on
$\mathcal G$, $J_1,J_2$, and Bob's string announced
before Alice measures in Line~5, Alice's outcomes on
$(J_1\cup J_2)\setminus\mathcal G$ are independent
uniform bits, and passing requires all these bits to
match Bob's reported bits, an event of probability
$2^{-|(J_1\cup J_2)\setminus\mathcal G|}$.
\end{enumerate}
 \end{proof}

The next lemma shows that replacing $A_j$ by
$\ket{0^m}\!\bra{0^m}$ in the real protocol
preserves the subnormalized final state on $VMR$.

\begin{lemma}
\label{lem:mal_bob_branch_replacement}
Fix the erasures, valid lists, completed tests,
and Clifford descriptions before the Bell
measurements in Line~6.
Suppose that \eqref{eq:mal_depolarizing} holds
for $\mathcal N_{V^{(j)}}$ with $Q=Q_j$ and
$\mathcal U=Q_j\setminus\mathcal G$.
Then, for every initial state $\rho^{A_1A_2ZR}$
and every compatible continuation of the public
communication, the real protocol produces the
same subnormalized state on $VMR$ from
$\rho^{A_1A_2ZR}$ and from
$\rho^{A_{3-j}ZR}\otimes
\ket{0^m}\!\bra{0^m}^{A_j}$,
using the same Bob strategy in both executions.
\end{lemma}

\begin{proof}
For an arbitrary $\rho^{A_1A_2ZR}$, define
$\rho'\triangleq\rho^{A_{3-j}ZR}\otimes
\ket{0^m}\!\bra{0^m}^{A_j}$, and fix the public
communication $M=\mu$.
Let $\rho_{\mathrm{real},\mathcal G,\mu}^{VMR}$
denote the subnormalized final state of the real
execution. Then, we have
\begin{align}
\rho_{\mathrm{real},\mathcal G,\mu}^{VMR}
(\rho,\mathcal D)
\nonumber
&\overset{(a)}{=}
\rho_{\mathrm{def},\mathcal G,\mu}^{VMR}
(\rho,\mathcal D)
\nonumber\\
&\overset{(b)}{=}
4^{-2m}
(\mathcal D_{\mathcal G,\mu}\otimes\operatorname{id}_R)
\left[
\operatorname{Tr}_{Y_{(Q_1\cup Q_2)\setminus\mathcal G}}
\left[
(\mathcal N_{V^{(1)}}\otimes
 \mathcal N_{V^{(2)}}\otimes\operatorname{id}_{ZR})
\left(
P_{T_1}P_{T_2}\rho
P_{T_2}^{\dagger}P_{T_1}^{\dagger}
\right)
\right]
\right]
\nonumber\\
&\overset{(c)}{=}
4^{-2m}
(\mathcal K_{\mathcal G,\mu}\otimes\operatorname{id}_R)
\left[
\operatorname{Tr}_{Y_{Q_j\setminus\mathcal G}}
\left[
(\mathcal N_{V^{(j)}}\otimes
 \operatorname{id}_{A_{3-j}ZR})
\left(P_{T_j}\rho P_{T_j}^{\dagger}\right)
\right]
\right]
\nonumber\\
&\overset{(d)}{=}
4^{-2m}
(\mathcal K_{\mathcal G,\mu}\otimes\operatorname{id}_R)
\left[
\frac{I_{Y_{Q_j\cap\mathcal G}}}{2^{|Q_j\cap\mathcal G|}}
\otimes\operatorname{Tr}_{A_j}
\left[P_{T_j}\rho P_{T_j}^{\dagger}\right]
\right]
\nonumber\\
&\overset{(e)}{=}
4^{-2m}
(\mathcal K_{\mathcal G,\mu}\otimes\operatorname{id}_R)
\left[
\frac{I_{Y_{Q_j\cap\mathcal G}}}{2^{|Q_j\cap\mathcal G|}}
\otimes\rho^{A_{3-j}ZR}
\right]
\nonumber\\\nonumber
&\overset{(f)}{=}
\rho_{\mathrm{real},\mathcal G,\mu}^{VMR}
(\rho',\mathcal D),
\end{align}
where
\begin{enumerate}[(a)]
\item holds with $\rho_{\mathrm{def},\mathcal G,\mu}^{VMR}$
defined as the subnormalized state obtained, for
fixed $\mu$, by evaluating Alice's tests in Line~5
and teleportation in Line~6 first, then applying
all of Bob's local maps in their original sequence.
This reordering preserves the state because
Alice's maps commute with Bob's maps, as the
two parties act on disjoint registers;
\item holds because, with Bob's maps deferred,
the tests in Line~5 leave the Bell pairs at
$Q_1\cup Q_2$ unchanged, and therefore applying 
\cite[Ex.~3.7.12]{wilde2013quantum} and
Lemma~\ref{lem:teleportation_properties_ach}
to these unchanged Bell pairs gives for each Bell outcome $T_i$ in Line~6,
the receiving map 
$\sigma^{A_i}\mapsto
4^{-m}\operatorname{Tr}_{Y_{Q_i\setminus\mathcal G}}
[\mathcal N_{V^{(i)}}(
P_{T_i}\sigma^{A_i}P_{T_i}^{\dagger})]$,
after tracing out Alice's unmeasured halves
and the erased   qubits.
Applying these maps to $A_1,A_2$, with
$\operatorname{id}_{ZR}$, supplies the input to
$\mathcal D_{\mathcal G,\mu}\otimes\operatorname{id}_R$,
where
$\mathcal D_{\mathcal G,\mu}:
\mathcal L(Y_{Q_1\cap\mathcal G}Y_{Q_2\cap\mathcal G}Z)
\longrightarrow\mathcal L(VM)$
tensors its input with the fixed subnormalized
state of $Y_{\mathcal G\setminus(Q_1\cup Q_2)}$
and the erasure flags after Alice's tests in
Line~5, then applies Bob's local maps in their
original sequence;

\item holds with
$\mathcal K_{\mathcal G,\mu}:
\mathcal L(Y_{Q_j\cap\mathcal G}A_{3-j}Z)
\longrightarrow\mathcal L(VM)$
defined by applying $P_{T_{3-j}}$ and
$\mathcal N_{V^{(3-j)}}$ to $A_{3-j}$,
tracing out $Y_{Q_{3-j}\setminus\mathcal G}$,
and then applying $\mathcal D_{\mathcal G,\mu}$.
Indeed, since $\mathcal N_{V^{(j)}}$ and
$\mathcal N_{V^{(3-j)}}$ have disjoint input
and output registers, we may encode $A_j$
and trace out $Y_{Q_j\setminus\mathcal G}$
before applying $\mathcal K_{\mathcal G,\mu}$;
\item holds  by the assumed  property
\eqref{eq:mal_depolarizing}, with $A_{3-j}ZR$
as the reference;
\item holds because $P_{T_j}$ acts only on the
traced-out register $A_j$;
\item holds by applying $(a)$--$(d)$ to $\rho'$,
whose marginal on $A_{3-j}ZR$ equals that of $\rho$.
\end{enumerate}
\end{proof}
Fix Bob's strategy $\mathcal D$.
In the ideal experiment, the simulator $\mathsf{Sim}_{\mathcal D}$ receives
$Z$ and runs $\mathcal D$ internally. Specifically, it prepares the Bell pairs and samples erasures
as in Line~1, performs Alice's operations in
Lines~2--5, and chooses $j$ 
as above.
If Line~6 is reached, it samples and announces
the Cliffords, chooses $C'=3-j$ and receives $A_{3-j}$, and completes the protocol with
$\ket{0^m}\!\bra{0^m}$ replacing $A_j$.
It outputs $\mathcal D$'s final register $V$
and the public communication~$M$. Then,
we have
\begin{align}
\sup_\rho\frac12
\left\|
\rho_{\mathrm{real}}^{HVMR}(\rho,\mathcal D)
-\rho_{\mathrm{ideal}}^{HVMR}
(\rho,\mathsf{Sim}_{\mathcal D})
\right\|_1
\nonumber
&\overset{(a)}{\leq}
\sup_\rho\Prb[\mathcal E]
\nonumber\\
&\overset{(b)}{\leq}
\Prb[|\mathcal G|>n(1-\epsilon)+r]
+64e^{-r^2/n}+2^{-2r^2/n}+4^{-r^2/n}
\nonumber\\
&\overset{(c)}{\leq}
e^{-2r^2/n}+64e^{-r^2/n}
+2^{-2r^2/n}+4^{-r^2/n}
\nonumber\\
&\overset{(d)}{\leq}
\frac{67\min\{\delta,\nu\}}{256}
\nonumber\\
&\leq\nu,
\label{eq:mal_bob_simulation}
\end{align}
where
\begin{enumerate}[(a)]
\item holds by Lemma~\ref{lem:mal_bob_branch_replacement}
outside $\mathcal E$ and by the triangle inequality
on $\mathcal E$, where $\mathcal E$ is the event that Line~6 is
reached and either
$|Q_j\setminus\mathcal G|<m+r^2/n$
or \eqref{eq:mal_depolarizing} fails for
$\mathcal N_{V^{(j)}}$ with $Q=Q_j$ and
$\mathcal U=Q_j\setminus\mathcal G$;

\item holds by
Lemma~\ref{lem:mal_remaining_erasure_bound}
and Lemma~\ref{lem:mal_clifford} because when $|Q_j\setminus\mathcal G|\geq m+r^2/n$,
\eqref{eq:mal_depolarizing}   fails
with probability at most
$4^{m-|Q_j\setminus\mathcal G|}
\leq4^{-r^2/n}$, uniformly over all inputs;

\item holds by Hoeffding's inequality for
$|\mathcal G|\sim\operatorname{Bin}(n,1-\epsilon)$;

\item holds because $e^{-2r^2/n}+64e^{-r^2/n}
+2^{-2r^2/n}+4^{-r^2/n}\leq \frac{67\min\{\delta,\nu\}}{256}$, using
$e^{-2r^2/n}\leq e^{-r^2/n}$ and
$2^{-2r^2/n}=4^{-r^2/n}\leq e^{-r^2/n}$, and  by the definition of $r$ in
\eqref{eq:mal_parameter_choice}.
\end{enumerate}
The same procedure $\mathsf{Sim}_{\mathcal D}$
works for every initial state $\rho$.
Thus, \eqref{eq:mal_bob_simulation} proves
\eqref{eq:qot_malicious_simulation}
for a dishonest Bob.

\subsubsection{Simulation against a dishonest Alice}

The ideal experiment requires the simulator to
generate $A'_1,A'_2,V,M$ from $Z$ without using $C$.
We construct the simulator from a modified
real execution that retains all transmitted
qubits $Y^n$ from Line~1 and omits Bob's decoding
in Line~7.
The following lemma shows that this execution
induces the same map from $Z$ to $VMY^n$ for
$C=1$ and $C=2$.

\begin{lemma}
\label{lem:mal_common_virtual_execution}
Fix Alice's strategy.
Modify the real execution by retaining all
transmitted qubits $Y^n$ from Line~1 and omitting
only Bob's decoding unitaries in Line~7.
Erasure sampling, public announcements, and
Alice's operations remain unchanged.
After tracing out Bob's private classical records,
the resulting map from $Z$ to Alice's final
private register $V$, the public communication
$M$, and $Y^n$ is identical for $C=1$ and $C=2$,
including aborts.
\end{lemma}

\begin{proof}
We first consider executions in which Bob
announces $a$ in Line~4. Fix accepted $k$,  lists $l_1,l_2$,
completed hashing messages $\gamma$ with outputs
$w_0,w_1$, and $a$ announced in Line~4.
Let $\mathcal V^{(c)}_{k,l_1,l_2,\gamma,a}$ denote
the map of the modified execution
through Bob's announcement of $a$ in Line~4,
with the public messages recorded and Bob's
private classical records traced out.
Then, we have
\begin{align}
\mathcal V^{(c)}_{k,l_1,l_2,\gamma,a}
&\overset{(a)}{=}
\frac{\Prb[|\mathcal G|=k]}
{(n)_{2t}\,2^{h_k}}\,
q_\gamma(w_{a\oplus(c-1)})
\mathcal K_{k,l_1,l_2,\gamma,a}
\nonumber\\
&\overset{(b)}{=}
\frac{\Prb[|\mathcal G|=k]}
{(n)_{2t}\,2^{h_k}}\,
q_\gamma(w_a)
\mathcal K_{k,l_1,l_2,\gamma,a}
\nonumber\\
&=
\mathcal V^{(1)}_{k,l_1,l_2,\gamma,a},
\label{eq:mal_virtual_hash_map}
\end{align}
where
\begin{enumerate}[(a)]
\item holds with $\mathcal K_{k,l_1,l_2,\gamma,a}$
defined by Alice's operations in Lines~1--4 for
the fixed messages, including preparation of
$Y^n$ and recording the messages.
Since no transmitted qubits are discarded,
the erasure pattern affects only the scalar
factors multiplying $\mathcal K_{k,l_1,l_2,\gamma,a}$, which are
$\Prb[|\mathcal G|=k]$ for the erasure count,
$1/(n)_{2t}$ for the lists by
\eqref{eq:mal_list_uniform},
$2^{-h_k}$ for Bob's uniform $W$ in Line~3,
and $q_\gamma(W)$ for his hashing replies,
where Line~4 fixes $W=w_{a\oplus(c-1)}$;

\item holds by \eqref{eq:mal_hash_equal_weights},
since $\mathcal K_{k,l_1,l_2,\gamma,a}$ depends
only on Alice's strategy and the fixed messages.
\end{enumerate}
We next extend \eqref{eq:mal_virtual_hash_map}
to the final registers $VMY^n$.
The fixed $k,l_1,l_2,\gamma,a$ determine
$J_1,J_2,Q_1,Q_2$, and Lines~3--4 give
$\mathcal G=J_1\cup J_2\cup Q_c$.
Thus, Bob measures $Y_{J_1\cup J_2}$ in the
computational basis in Line~5 for both
$C=1$ and $C=2$.
For fixed later messages, Alice's local maps,
Bob's message checks, and his announcement
of $F$ also do not use $C$.
Equality therefore holds at the end of every
continuation, including aborts.

Finally, for executions that abort before $a$ is announced,
the erasure count in Line~2 and list sampling
in Line~3 have distributions independent of $C$.
Before Line~4, \eqref{eq:mal_list_uniform}
also makes $W$ uniform conditioned on
$k,L_1,L_2$, and Bob's hashing operations use
$W$ but not $C$.
Hence, each fixed public record of an early
abort induces the same subnormalized map
from $Z$ to $VMY^n$ for $C=1$ and $C=2$.

\end{proof}
For the same initial state and fixed choice $C=c$,
the real final state on $B_{\mathrm{out}}VMR$
is recovered from the modified execution as follows.
If $F=0$, apply Line~7's decoding to $Y_{Q_c}$
and discard all transmitted qubits except Bob's
$m$-qubit output.
If $F=1$, discard all transmitted qubits and
place $\bot$ in $B_{\mathrm{out}}$.
Indeed, $Q_c\subseteq\mathcal G$ on success.
The trace over erased qubits can be postponed
because honest Bob never acts on them, and
his decoding can be postponed because it
commutes with Alice's local operations.

Fix Alice's strategy $\mathcal D$.
In the ideal experiment, $\mathsf{Sim}_{\mathcal D}$
receives $Z$ and internally runs the modified
real execution of
Lemma~\ref{lem:mal_common_virtual_execution},
using $\mathcal D$ for Alice's operations and
the fixed choice $c=1$ for Bob's operations. For $F=0$, it decodes each $Y_{Q_i}$,
$i\in\{1,2\}$, as in Line~7:
it applies $(V^{(i)})^\dagger$, followed by
$P_{T_i}^\dagger$ on the first $m$ qubits,
which it retains as $A'_i$.
It then discards the other transmitted qubits.
For $F=1$, it discards all transmitted qubits
and prepares fixed states in $A'_1,A'_2$. The simulator gives $A'_1,A'_2$ to the genie
and outputs $\mathcal D$'s final private register
$V$ and the generated public communication $M$.
Its operations define a CPTP map
$\mathcal A:\mathcal L(Z)\to
\mathcal L(A'_1A'_2VM)$ that does not use $C$.

For $f\in\{0,1\}$, define
$\mathcal A_f(\sigma)
\triangleq
(I\otimes\ket f\!\bra f^F)\mathcal A(\sigma)
(I\otimes\ket f\!\bra f^F)$, then, for every initial state
$\rho^{CZR}=\frac12\sum_{c=1}^{2}
\ket c\!\bra c^C\otimes\rho_c^{ZR}$,
we have
\begin{align}
\rho_{\mathrm{real}}^{HVMR}(\rho,\mathcal D)
\nonumber
&\overset{(a)}{=}
\frac12\sum_{c=1}^{2}
\ket c\!\bra c^C\otimes
\Bigg[
(\operatorname{id}_{A'_c\to B_{\mathrm{out}}}
\otimes\operatorname{id}_{VMR})
\operatorname{Tr}_{A'_{\bar c}}
\left[
(\mathcal A_0\otimes\operatorname{id}_R)
(\rho_c^{ZR})
\right]
\nonumber\\
&\qquad\qquad+
\ket{\bot}\!\bra{\bot}^{B_{\mathrm{out}}}
\otimes\operatorname{Tr}_{A'_1A'_2}
\left[
(\mathcal A_1\otimes\operatorname{id}_R)
(\rho_c^{ZR})
\right]
\Bigg]
\nonumber\\
&\overset{(b)}{=}
\rho_{\mathrm{ideal}}^{HVMR}
(\rho,\mathsf{Sim}_{\mathcal D}),
\label{eq:mal_alice_output}
\end{align}
where
\begin{enumerate}[(a)]
\item holds because, for the same input $\rho_c^{ZR}$,
Lemma~\ref{lem:mal_common_virtual_execution}
gives identical states on $VMY^nR$ in the modified
execution with the simulator's internal choice $1$
and with choice $c$;

\item holds by the definition of the ideal experiment.
\end{enumerate}

The same simulator works for every allowed
initial state $\rho$, thus \eqref{eq:mal_alice_output} proves
\eqref{eq:qot_malicious_simulation} for a
dishonest Alice.
\section{Concluding Remarks}
\label{sec:concl}

We studied quantum-state oblivious transfer over the quantum erasure channel in honest-but-curious and malicious settings, with no initial shared entanglement. We established  finite-blocklength converse   and achievability bounds. The honest-but-curious bound is exact at finite blocklength, and both settings have capacity $\min\{\epsilon,1-\epsilon\}$. The achievability scheme shows that unerased channel uses can be used to teleport the chosen quantum state, while erased channel uses protect the non-chosen state. The converse shows that these two resources impose complementary limitations: Bob must receive enough quantum systems to reconstruct the chosen state, but enough systems must also be unavailable to Bob to preserve privacy of the non-chosen state.

\appendices

\section{Proof of Corollary \ref{cor:second_order_fixed_delta_achievability}}
\label{app_sec_ac}

First, suppose that $0<\epsilon<1/2$. Fix a constant $K>0$ to be
chosen, and define
\begin{align}
m_n^-
\triangleq
\left\lfloor
n\epsilon
+
\sqrt{n\epsilon(1-\epsilon)}\,\Phi^{-1}(\delta)
-K
\right\rfloor,
 \quad m_n^+
\triangleq
\left\lceil
n\epsilon
+
\sqrt{n\epsilon(1-\epsilon)}\,\Phi^{-1}(\delta)
+K
\right\rceil.
\label{eqmnplus}
\end{align}
We have
\begin{align*}
\Prb[\operatorname{Bin}(n,\epsilon)<m_n^\pm]
&=
\Prb\left[
\frac{\operatorname{Bin}(n,\epsilon)-n\epsilon}
{\sqrt{n\epsilon(1-\epsilon)}}
<
\frac{m_n^\pm-n\epsilon}
{\sqrt{n\epsilon(1-\epsilon)}}
\right]\\
&\overset{(a)}{=}
\Phi\left(
\frac{m_n^\pm-n\epsilon}
{\sqrt{n\epsilon(1-\epsilon)}}
\right)
+
O(n^{-1/2})\\
&\overset{(b)}{=}
\Phi\left(
\Phi^{-1}(\delta)
\pm
\frac{K}{\sqrt{n\epsilon(1-\epsilon)}}
+
O(n^{-1/2})
\right)
+
O(n^{-1/2})\\
&\overset{(c)}{=}
\delta
\pm
\frac{
K\phi(\Phi^{-1}(\delta))
}{
\sqrt{n\epsilon(1-\epsilon)}
}
+
O(n^{-1/2}), \numberthis\label{eqcor1a}
\end{align*}
where $(a)$ holds because  by the Berry--Esseen theorem, e.g., \cite[Th.~2]{shevtsova2011absolute}, for every  $x \in \mathbb{R}$,
\begin{align*}
\left|
\Prb\left[
\frac{\operatorname{Bin}(n,\epsilon)-n\epsilon}
{\sqrt{n\epsilon(1-\epsilon)}}<x
\right]
-\Phi(x)
\right|
\leq
\frac{0.33554}{\sqrt n}
\left(
\frac{\epsilon^2+(1-\epsilon)^2}
{\sqrt{\epsilon(1-\epsilon)}}
+0.415
\right),
\end{align*}
$(b)$ holds by \eqref{eqmnplus}, $(c)$ follows from the first-order Taylor expansion
$
\Phi\big(\Phi^{-1}(\delta)+h\big)
=
\Phi\big(\Phi^{-1}(\delta)\big)
+
\phi\big(\Phi^{-1}(\delta)\big)h
+
O(h^2)
=
\delta
+
\phi\big(\Phi^{-1}(\delta)\big)h
+
O(h^2),$ 
applied with
$
h=\pm K/\sqrt{n\epsilon(1-\epsilon)}+O(n^{-1/2})
$, and where $\phi$ denotes the standard Gaussian probability density function.

Then, choose $K$ sufficiently large and then choose $n$
sufficiently large so that, for some  $c>0$, we have from~\eqref{eqcor1a}
\begin{align}
\Prb[\operatorname{Bin}(n,\epsilon)<m_n^-]
\leq
\delta-\frac{c}{\sqrt n}, \quad
\Prb[\operatorname{Bin}(n,\epsilon)<m_n^+]
\geq
\delta+\frac{c}{\sqrt n}.\label{eqcor1c}
\end{align}

Then, for all sufficiently large $n$,
\begin{align*}
\Prb[
m_n^-\leq\operatorname{Bin}(n,\epsilon)\leq n-m_n^-
]
&=
1-
\Prb[\operatorname{Bin}(n,\epsilon)<m_n^-]
-
\Prb[\operatorname{Bin}(n,\epsilon)>n-m_n^-]\\
&\overset{(a)}{\geq}
1-\delta+\frac{c}{\sqrt n}-\Prb[\operatorname{Bin}(n,\epsilon)>n-m_n^-]\\
&\overset{(b)}{\geq}
1-\delta+\frac{c}{\sqrt n}-\Prb[\operatorname{Bin}(n,\epsilon)\geq n/2]\\
&\overset{(c)}{\geq}
1-\delta+\frac{c}{\sqrt n}-\exp\left\{
-n\left(
\frac12\ln\frac{1/2}{\epsilon}
+
\frac12\ln\frac{1/2}{1-\epsilon}
\right)
\right\}\\
&\overset{(d)}{\geq}
1-\delta+\frac{c}{\sqrt n}-e^{-c'n}\\
&\overset{(e)}{\geq}
1-\delta, \numberthis\label{eqcor1f1}
\end{align*}
where
$(a)$ holds by \eqref{eqcor1c}, $(b)$ holds because for all sufficiently large $n$, $
n-m_n^->n/2
$ since $\epsilon<1/2$ and $m_n^-=n\epsilon+O(\sqrt n)$, $(c)$ holds by a Chernoff
 bound, e.g., \cite[Th.~2.1]{mulzer2018five}, $(d)$ holds with
$c'\triangleq \frac12\ln\frac{1/2}{\epsilon}
+
\frac12\ln\frac{1/2}{1-\epsilon}
=
-\frac12\ln\big(4\epsilon(1-\epsilon)\big)
>
0
$ 
since $\epsilon<1/2$, $(e)$ holds
because $e^{-c'n}=o(n^{-1/2})$.

On the other hand, for all sufficiently large $n$,
\begin{align*}
\Prb[
m_n^+\leq\operatorname{Bin}(n,\epsilon)\leq n-m_n^+
]
&\leq
1-
\Prb[\operatorname{Bin}(n,\epsilon)<m_n^+]\\
&\overset{(a)}{\leq}
1-\delta-\frac{c}{\sqrt n}\\
&\overset{(b)}{<}
1-\delta, \numberthis\label{eqcor1f2}
\end{align*}
where $(a)$ holds by \eqref{eqcor1c}, and $(b)$ holds because $c>0$.

Then, by \eqref{eqcor1f1}, \eqref{eqcor1f2} and Theorem~\ref{thm:finite_length_achievability},
$m_n^-
\leq
M_{\mathrm{QOT}}(\epsilon,n,\delta)
<
m_n^+$, 
where  \eqref{eqmnplus} shows that $m_n^-$ and $m_n^+$
differ from
$
n\epsilon+
\sqrt{n\epsilon(1-\epsilon)}\,\Phi^{-1}(\delta)
$
by at most a constant. Hence,
$
M_{\mathrm{QOT}}(\epsilon,n,\delta)
=
n\epsilon
+
\sqrt{n\epsilon(1-\epsilon)}\,\Phi^{-1}(\delta)
+
O(1).$

Next, suppose that $1/2<\epsilon<1$. Since
$n- \operatorname{Bin}(n,\epsilon)
$
 has the same distribution as
$
 \operatorname{Bin}(n,1-\epsilon) ,
$
applying the previous argument to $\operatorname{Bin}(n,1-\epsilon)$,
where $1-\epsilon<1/2$, gives
 $M_{\mathrm{QOT}}(\epsilon,n,\delta)
=
n(1-\epsilon)
+
\sqrt{n\epsilon(1-\epsilon)}\,\Phi^{-1}(\delta)
+
O(1).
$

Finally, consider $\epsilon=1/2$. Fix a constant $K>0$ to be chosen,
and define
\begin{align}
m_{n,1/2}^-
\triangleq
\left\lfloor
\frac n2
+
\frac{\sqrt n}{2}
\Phi^{-1}\left(\frac{\delta}{2}\right)
-K
\right\rfloor,
\quad 
m_{n,1/2}^+
\triangleq
\left\lceil
\frac n2
+
\frac{\sqrt n}{2}
\Phi^{-1}\left(\frac{\delta}{2}\right)
+K
\right\rceil.
\label{eq:mnhalfplus}
\end{align}
Applying the  argument in
\eqref{eqcor1a} with $\epsilon=1/2$ and $\delta$ replaced by
$\delta/2$, and then choosing $K$ and $n$ sufficiently large as in
  \eqref{eqcor1c}, give, for some constant $c>0$,
\begin{align}
\Prb\left[
\operatorname{Bin}\left(n,\frac12\right)<m_{n,1/2}^-
\right]
\leq
\frac{\delta}{2}-\frac{c}{\sqrt n},
\label{eq:half_lower_quantile}\quad 
\Prb\left[
\operatorname{Bin}\left(n,\frac12\right)<m_{n,1/2}^+
\right]
\geq
\frac{\delta}{2}+\frac{c}{\sqrt n}.
\end{align}
Then, for all sufficiently large
$n$, we have
\begin{align*}
\Prb\left[
m_{n,1/2}^-
\leq
\operatorname{Bin}\left(n,\frac12\right)
\leq
n-m_{n,1/2}^-
\right]
& = 1-
\Prb\left[
\operatorname{Bin}\left(n,\frac12\right)<m_{n,1/2}^-
\right]
-
\Prb\left[
\operatorname{Bin}\left(n,\frac12\right)>n-m_{n,1/2}^-
\right]\\
& = 1-
\Prb\left[
\operatorname{Bin}\left(n,\frac12\right)<m_{n,1/2}^-
\right]
-
\Prb\left[
n-\operatorname{Bin}\left(n,\frac12\right)<m_{n,1/2}^-
\right]\\
&\overset{(a)}{=}
1-
2\Prb\left[
\operatorname{Bin}\left(n,\frac12\right)<m_{n,1/2}^-
\right]\\
&\overset{(b)}{\geq}
1-\delta+\frac{2c}{\sqrt n}\\
&\geq
1-\delta,
\numberthis
\label{eq:half_achievable}
\end{align*}
where $(a)$ holds because
$
n-\operatorname{Bin}(n,1/2)
$
has the same distribution as
$
\operatorname{Bin}(n,1/2)
$, $(b)$
follows from \eqref{eq:half_lower_quantile}. Similarly, 
\begin{align}
\Prb\left[
m_{n,1/2}^+
\leq
\operatorname{Bin}\left(n,\frac12\right)
\leq
n-m_{n,1/2}^+
\right]
<
1-\delta.
\numberthis
\label{eq:half_converse}
\end{align}

Hence, by \eqref{eq:half_achievable},
\eqref{eq:half_converse}, and
Theorem~\ref{thm:finite_length_achievability}, we have
$m_{n,1/2}^-
\leq
M_{\mathrm{QOT}}\left(\frac12,n,\delta\right)
<
m_{n,1/2}^+$,
and by \eqref{eq:mnhalfplus}, we conclude that
$
M_{\mathrm{QOT}}\left(\frac12,n,\delta\right)
=
\frac n2
+
\frac{\sqrt n}{2}
\Phi^{-1}\left(\frac{\delta}{2}\right)
+
O(1).$

\section{Proof of Lemma \ref{lem:schmidt_rank_growth}} \label{app_B}

\eqref{eq:sr_initial_one} holds because Alice and Bob have no
initial shared state. For each $i$, we condition only on the erasure outcomes and public
messages produced up to that time. For fixed values of these public
messages, every local operation or public-message exchange is represented
by a local linear map together with a fixed classical basis state at the
receiving party and therefore cannot increase Schmidt rank.  
Hence, the only operation that can increase the Schmidt rank across
Bob's laboratory versus the remaining systems is an unerased quantum
transmission.

Suppose immediately before an unerased transmission the pure state is on
$\mathcal L X\mathcal B$, where $X$ is the qubit to be delivered to Bob,
$\mathcal B$ is Bob's current register, and $\mathcal L$ is the current collection of registers outside Bob’s laboratory. If
$
\operatorname{SR}_{\mathcal L X:\mathcal B}(\psi)=r,
$
then
\begin{align*}
|\psi\rangle^{\mathcal L X\mathcal B}
=
\sum_{j=1}^{r}
|\alpha_j\rangle^{\mathcal L X}
|\beta_j\rangle^{\mathcal B}
=
\sum_{j=1}^{r}
\sum_{x\in\{0,1\}}
|\alpha_{j,x}\rangle^{\mathcal L}
\left(
|x\rangle^X|\beta_j\rangle^{\mathcal B}
\right),
\end{align*}
where  the last equality holds using the computational basis of the qubit $X$,
$
|\alpha_j\rangle^{\mathcal L X}
=
|\alpha_{j,0}\rangle^{\mathcal L}|0\rangle^X
+
|\alpha_{j,1}\rangle^{\mathcal L}|1\rangle^X.
$
This is a decomposition into at most $2r$ product terms across
$\mathcal L:X\mathcal B$. Hence, an unerased qubit transmission increases
Schmidt rank by at most a factor of $2$.

If the $i$-th channel use is erased, Bob receives only the erasure flag,
while the quantum input is placed in the environment $W_i$, which is part of $\mathcal L_i$, and the Schmidt rank does
not increase. Hence,
$
R_i
\leq
\begin{cases}
2R_{i-1},&z_i=1,\\
R_{i-1},&z_i=0,
\end{cases}
$
which gives \eqref{eq:sr_one_step_growth}.

\bibliographystyle{IEEEtran}
\bibliography{bib}

\end{document}